\documentclass[letterpaper]{article}
\usepackage{amsmath}
\usepackage{graphicx}
\usepackage{subcaption}
\usepackage{caption}
\usepackage{amssymb}
\usepackage{amsthm}
\usepackage{color}
\usepackage{hyperref}
\usepackage{cleveref}
\usepackage{float}
\usepackage{mathtools}

\usepackage[margin=1in]{geometry}

\hypersetup{hidelinks,
	colorlinks=true,
	allcolors=black,
	pdfstartview=Fit,
	breaklinks=true}

\crefname{theorem}{Theorem}{Theorem}
\crefname{lemma}{Lemma}{Lemma}
\crefname{definition}{Definition}{Definition}
\crefname{figure}{Figure}{Figure}
\crefname{section}{Section}{Section}

\newtheorem{theorem}{Theorem}
\newtheorem{definition}[theorem]{Definition}
\newtheorem{lemma}[theorem]{Lemma}
\newtheorem{proposition}[theorem]{Proposition}

\crefname{equation}{}{}

\title{Approximability of the Four-Vertex Model}
\author{Zhiguo Fu\footnote{School of Information Science and Technology and KLAS, Northeast Normal University, Changchun, China. Email: \href{mailto:fuzg432@nenu.edu.cn}{fuzg432@nenu.edu.cn}}\qquad
Tianyu Liu\footnote{Stanford University, Stanford, CA, USA. Email: \href{mailto:tliu248@stanford.edu}{tliu248@stanford.edu}}\qquad
Xiongxin Yang \footnote{School of Information Science and Technology, Northeast Normal University, Changchun, China. Email: \href{mailto:yangxx500@nenu.edu.cn}{yangxx500@nenu.edu.cn}}}
\date{}

\begin{document}

\maketitle

\begin{abstract}
    We study the approximability of the four-vertex model, a special case of the six-vertex model.
    We prove that, despite being NP-hard to approximate in the worst case, the four-vertex model admits a fully polynomial randomized approximation scheme (FPRAS) when the input satisfies certain linear equation system over GF($2$).
    The FPRAS is given by a Markov chain known as the \emph{worm process}, whose state space and rapid mixing rely on the solution of the linear equation system. 
    This is the first attempt to design an FPRAS for the six-vertex model with \emph{unwindable} constraint functions.
    Additionally, we explore the applications of this technique on planar graphs, providing efficient sampling algorithms.
\end{abstract}

\section{Introduction}
\label{Introduction}

The \textit{six-vertex model} was first introduced by Linus Pauling \cite{Pauling/1935/Ice_type} in 1935 to describe the properties of ice. 
It is an abstraction of the crystal lattice with hydrogen bonds.
From a graph-theoretic perspective, the six-vertex model is defined on a 4-regular graph $G$ as follows.
Suppose the four incident edges of each vertex are labeled from $1$ to $4$.
The state space of the six-vertex model consist of all the Eulerian orientations of $G$, i.e., orientations where each vertex has exactly two arrows coming in and two arrows going out.
The ``two-in two-out'' rule is also called the ice-rule. 
There are $\binom{4}{2} = 6 $ permitted types of local configurations around a vertex hence the name six-vertex model (see \cref{six-type}).

\begin{figure}[ht]
    \begin{subfigure}[t]{0.15\textwidth}
        \centering
        \includegraphics[width=1\textwidth]{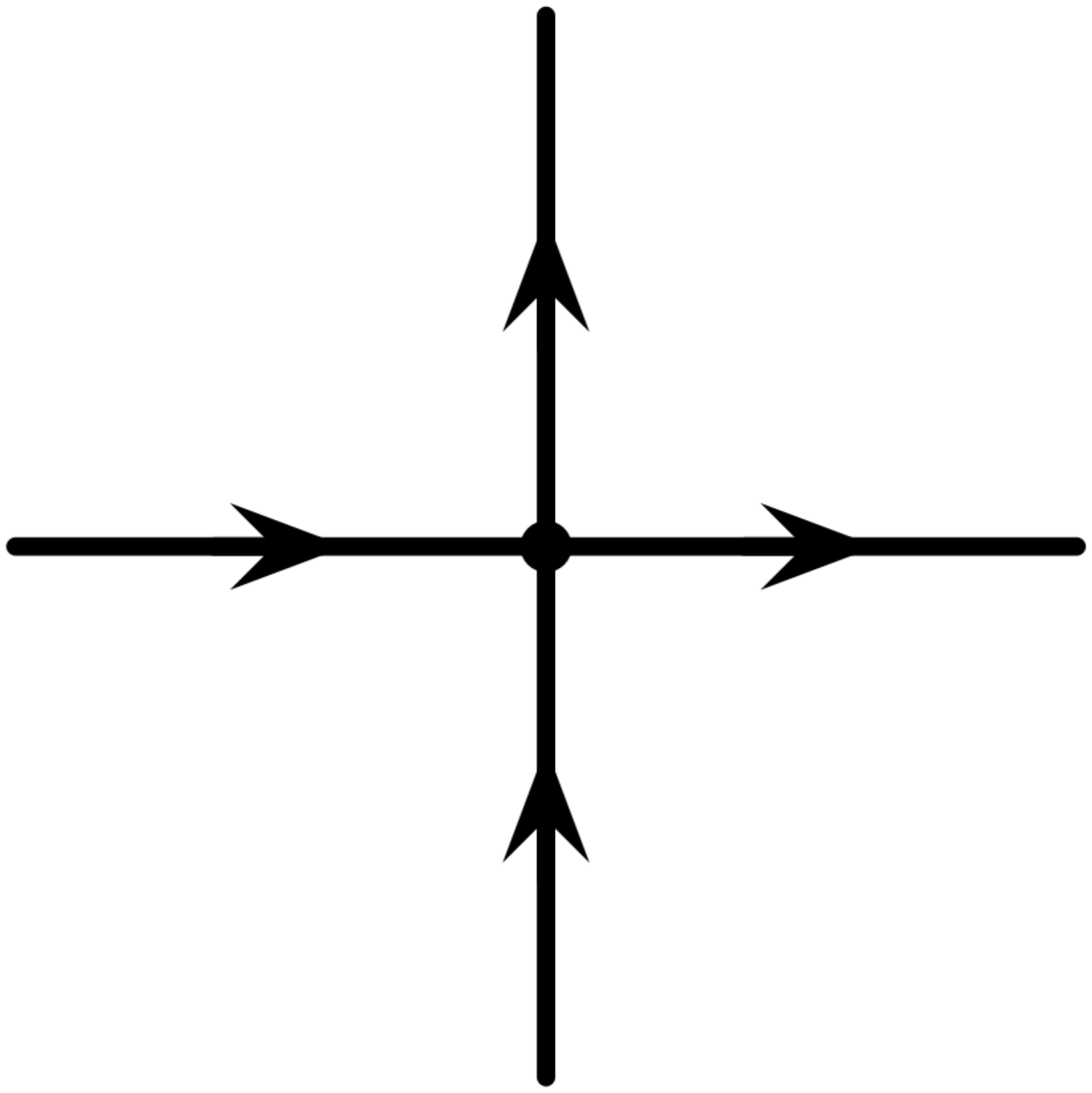}
        \subcaption*{1}
    \end{subfigure}
    \begin{subfigure}[t]{0.15\textwidth}
        \centering
        \includegraphics[width=1\textwidth]{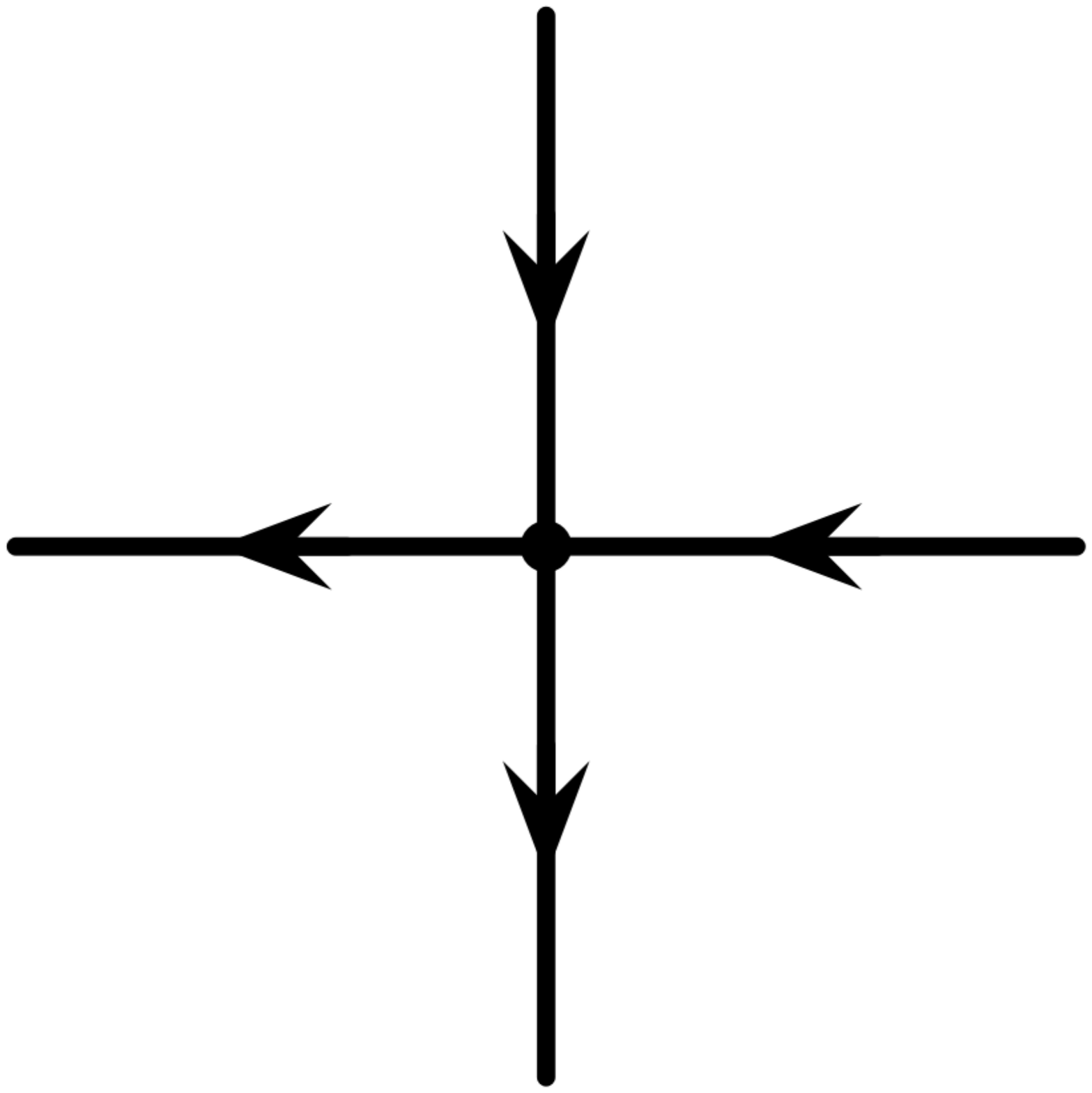}
        \subcaption*{2}
    \end{subfigure}
    \begin{subfigure}[t]{0.15\textwidth}
        \centering
        \includegraphics[width=1\textwidth]{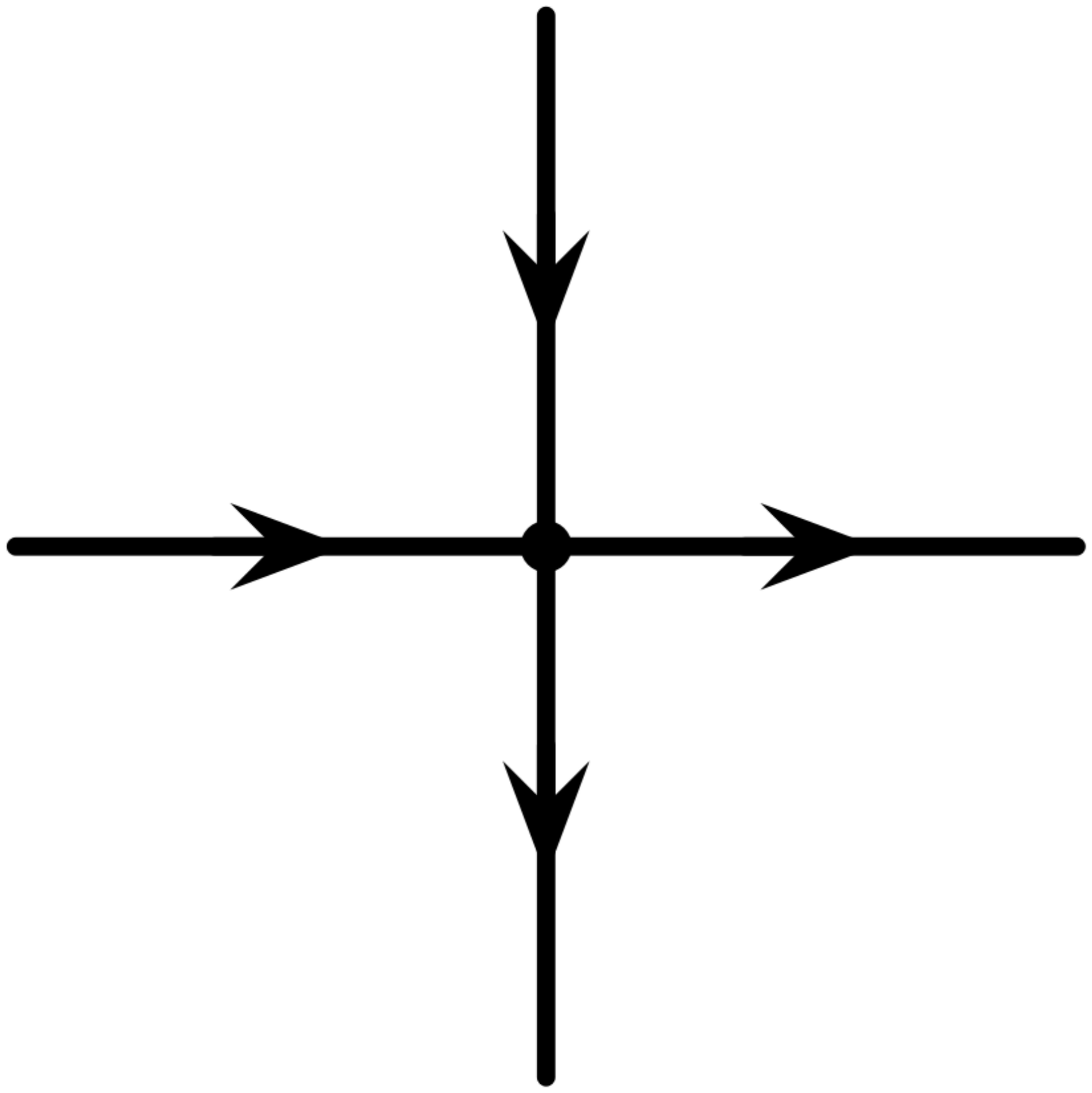}
        \subcaption*{3}
    \end{subfigure}
    \begin{subfigure}[t]{0.15\textwidth}
        \centering
        \includegraphics[width=1\textwidth]{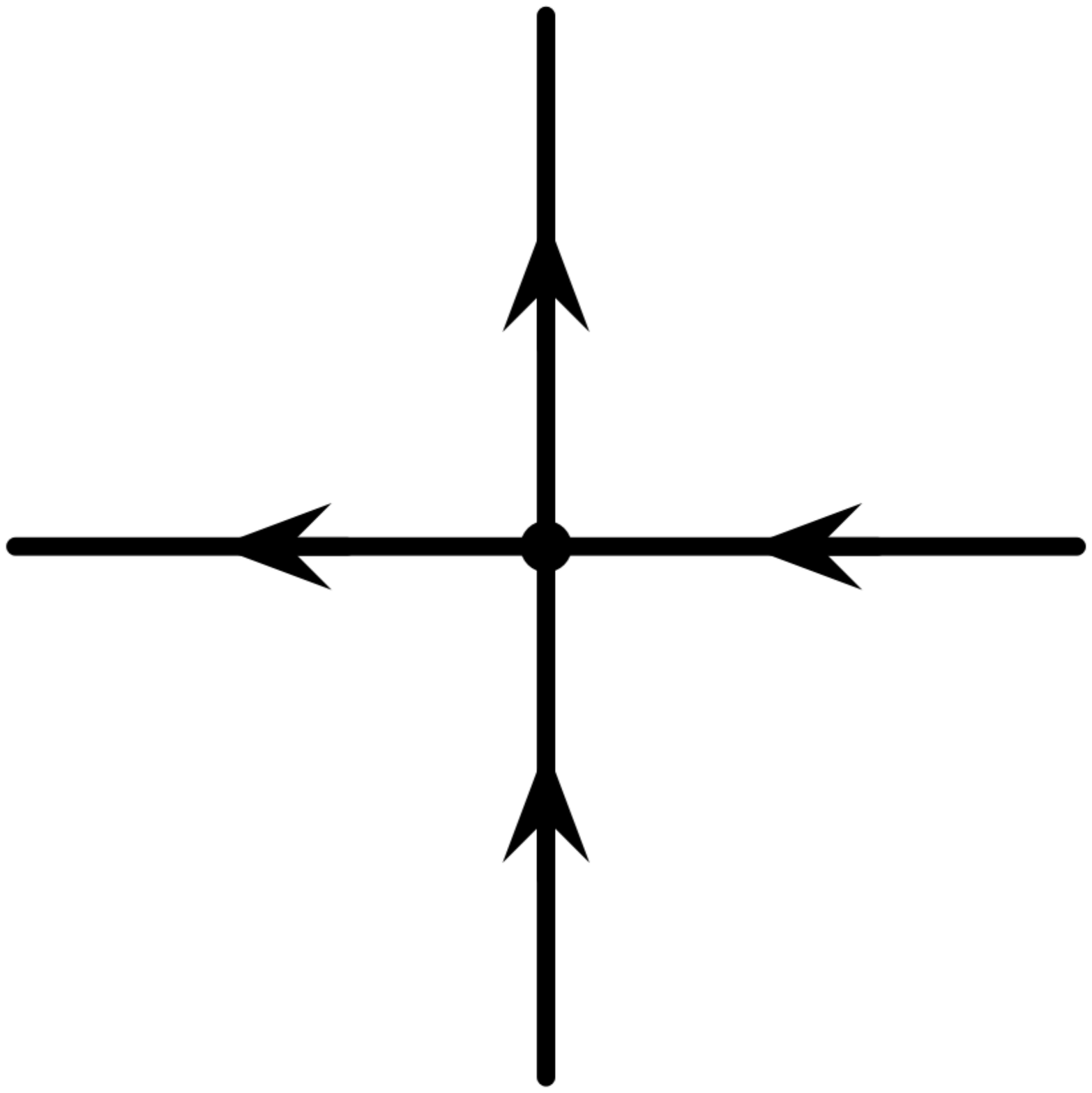}
        \subcaption*{4}
    \end{subfigure}
    \begin{subfigure}[t]{0.15\textwidth}
        \centering
        \includegraphics[width=1\textwidth]{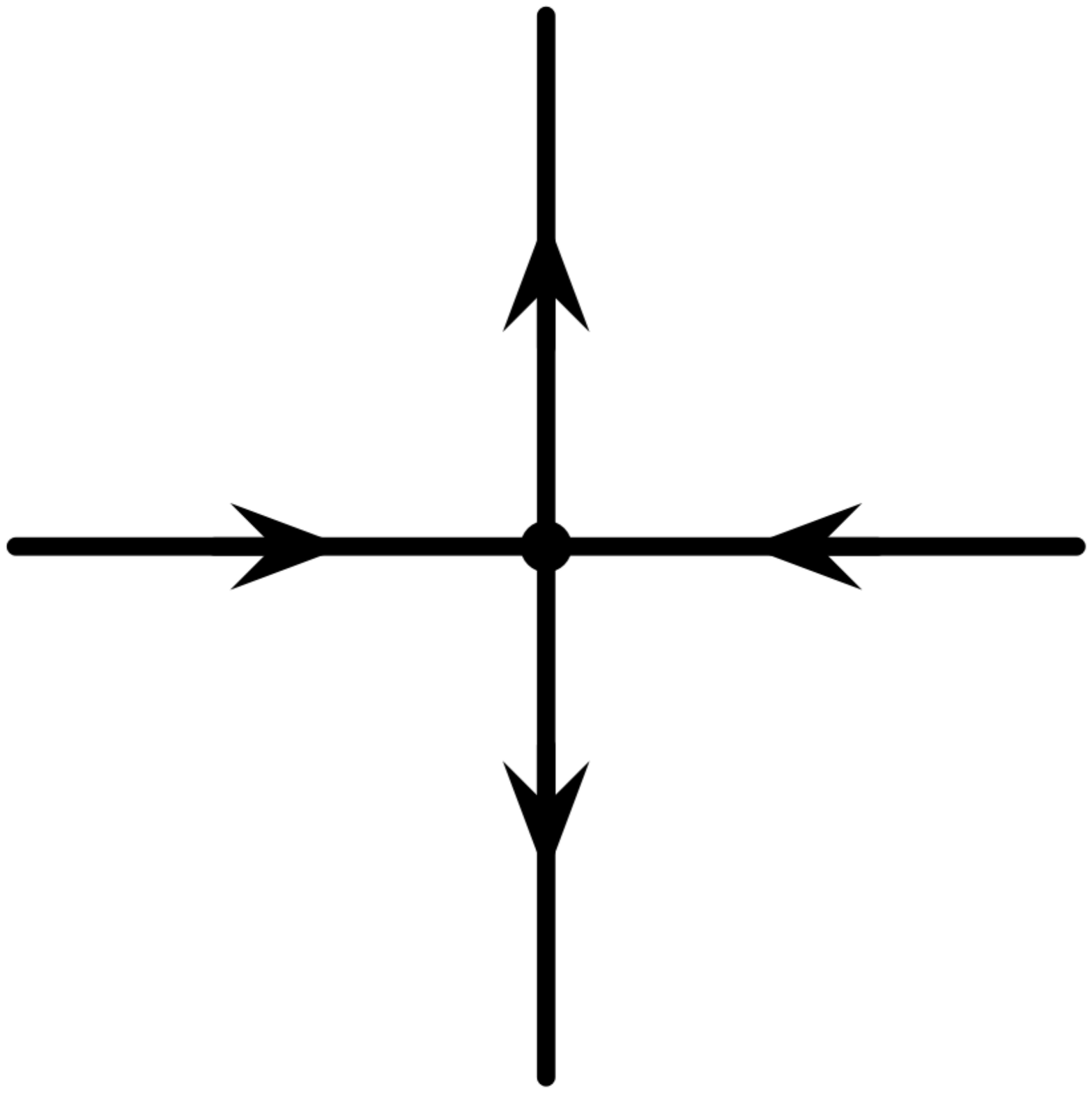}
        \subcaption*{5}
    \end{subfigure}
    \begin{subfigure}[t]{0.15\textwidth}
        \centering
        \includegraphics[width=1\textwidth]{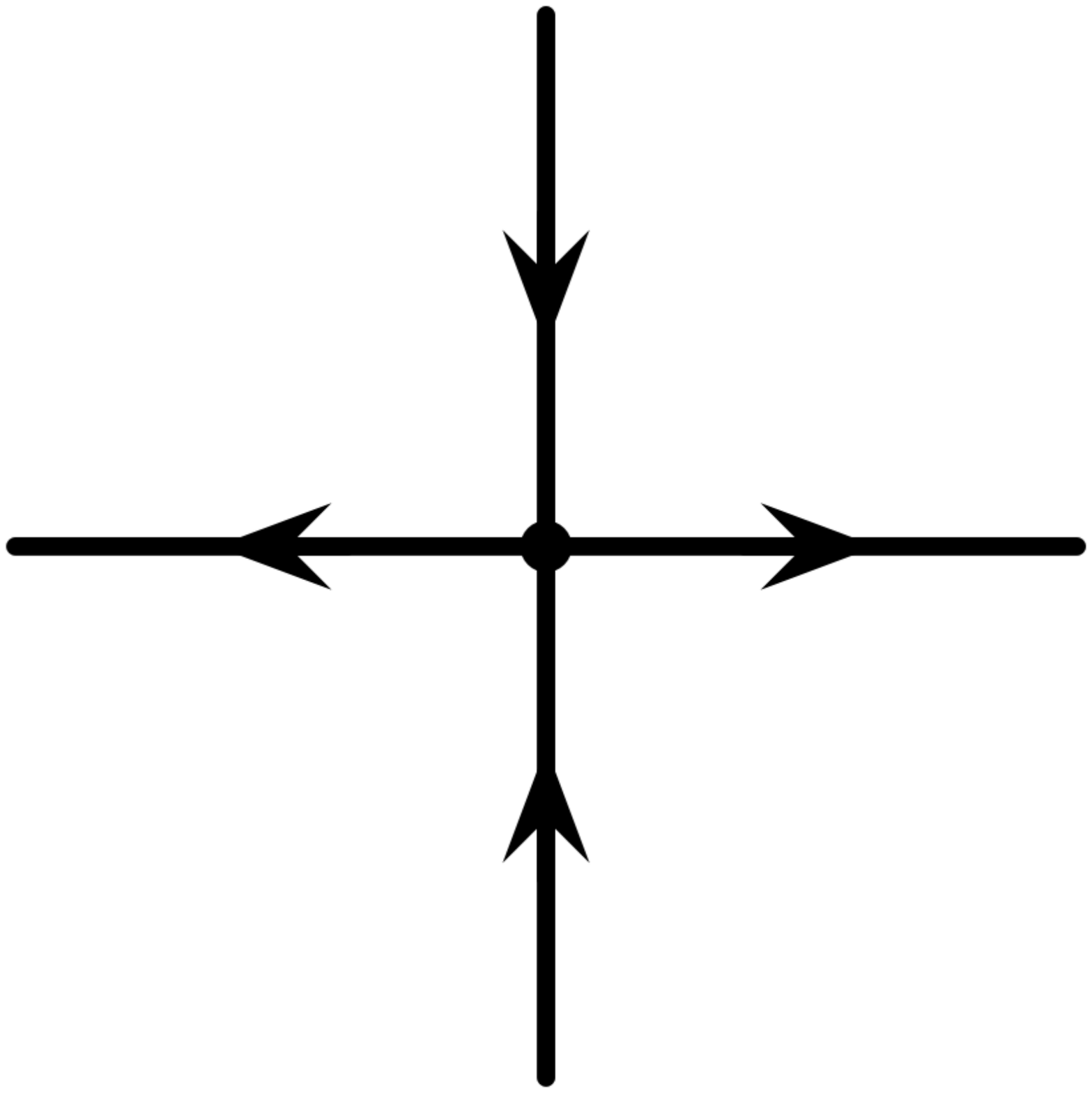}
        \subcaption*{6}
    \end{subfigure}
    \caption{Valid configurations of the six-vertex model}
    \label{six-type}
\end{figure}

In general, the six configurations $1$ to $6$ in \cref{six-type} are associated with six weights $\omega_1,\omega_2,\cdots,\omega_6$.
In this paper, we assume the \textit{arrow reversal symmetry}, i.e., $\omega_1=\omega_2=a, \omega_3=\omega_4=b, \omega_5=\omega_6=c$ and $a,b,c\ge 0$ as per the convention in physics.
The \emph{partition function} of the six-vertex model with parameters $(a,b,c)$ on a 4-regular graph $G$ is defined as
\begin{equation*}
    Z_{6V}(G;a,b,c)\coloneqq\sum_{\tau\in\mathcal{EO}(G)}a^{n_1+n_2}b^{n_3+n_4}c^{n_5+n_6},
\end{equation*}
where $\mathcal{EO}(G)$ is the set of all Eulerian orientations of $G$ and $n_i$ is the number of vertices of type $i (1\le i\le 6)$ in the graph under an Eulerian orientation $\tau\in\mathcal{EO}(G)$.

Interestingly, in physics the six-vertex model exhibits phase transition phenomena when the parameters $a, b, c$ vary (as a result of temperature change).
When the triangle inequalities hold ($a \le b + c$, $b \le a + c$, and $c \le a + b$), the six-vertex model is in the \emph{high-temperature regime}; otherwise, it is in the \emph{low-temperature regime}. 
For details, please refer to \cite{Baxter/book}.
Beyond physics, researchers have discovered many connections between the six-vertex model and other areas, like the alternating sign matrix (ASM) conjecture \cite{Kuperberg/1996/Another_ASM} and the Tutte polynomial\cite{Tutte/1954/Chromatic_Polynomials, Las/1988/Evaluation_Tutte}.
Furthermore, it can be extended to the eight-vertex model,
which has a close connection to other important models in statistical physics such as zero-field Ising model \cite{Baxter/1972/Partition_Eight_Vertex}.

As a sum-of-product computation, the partition function of the six-vertex model is a counting problem, and can be expressed as a family of \emph{Holant problems}.
Holant problem is a framework for studying counting problem and is expressive enough to contain classical frameworks such as \emph{Graph Homomorphism} (GH) and \emph{Counting Constraint Satisfaction Problems} (\#CSP) as special cases.
A series of theorems of complexity classification were built for GH and \#CSP for exact computation 
\cite{Goldberg/Grohe/Jerrum/Thurley/2010/Complexity_Dichotomy_Mixed_Signs,
Bulatov/Dyer/Goldberg/Jalsenius/Jerrum/Richerby/2012/Complexity_Counting_CSP,
Bulatov/2013/Complexity_Counting_CSP,
Heng_Guo/Williams/2013/Complexity_Counting_CSP_Complex,
Dyer/Richerby/2013/Dichotomy_Counting_CSP,
J.Cai/X.Chen/P.Lu/2013/GH_Complex,
J.Cai/H.Guo/T.Williams/2016/Vanishing_Signatures,
S.Huang/P.Lu/2016/Real_Holant,
J.Cai/X.Chen/2017/Complexity_Complex_Counting_CSP,
J.Lin/H.Wang/2018/Boolean_Holant,
M.Backens/2018/Complex_Holant,
J.Cai/Z.Fu/S.Shao/2021/Planar_Six_Vertex}  and approximation
\cite{M.Jerrum/A.Sinclair/1989/Approximating_Permanent,
A.Sinclair/1992/Mixing_FLow,
Jerrum/Sinclair/1993/Approximating_Ising,
D.Randall/P.Tetali/1998/Glauber_Comparison,
M.Luby/D.Randall/A.Sinclair/2001/MC_Planar_Lattice,
L.Goldberg/R.Martin/M.Paterson/2004/Sample_3coloring,
M.Jerrum/A.Sinclair/E.Vigoda/2004/Approximation_Permanent_Nonnegative,
Goldberg/Mark/2012/Approximating_Potts}.
However, results are very limited for Holant problems, in particular for their approximate complexity.

The six-vertex model is an important base case for studying Holant problems with asymmetric constraint functions since it cannot be expressed by GH or \#CSP~\cite{J.Cai/Z.Fu/M.Xia/2018/Complexity_Six_Vertex}.
For the exact computation, a complexity dichotomy was proved by Cai et al. \cite{J.Cai/Z.Fu/M.Xia/2018/Complexity_Six_Vertex}.

For the approximate complexity, Mihail and Winkler \cite{M.Mihail/P.Winkler/1996/EO} gave the first fully polynomial randomized approximation scheme (FPRAS) in the unweighted case, i.e., counting the Eulerian orientations of the underlying graph.
In the weighted case, Cai et al. \cite{J.Cai/T.Liu/P.Lu/2019/Approximability_Six_Vertex} showed that the approximability of the six-vertex model is dramatically similar to the phase transition phenomenon in statistical physics as the temperature varies.
They proved that there is no FPRAS in the entire low-temperature regime unless RP = NP, whereas an FPRAS was only given in a subregion of the high-temperature regime via \emph{Markov chain Monte Carlo(MCMC)}.
In particular, the parameter settings in this subregion are ``windable'' -- a notion proposed by McQuillan~\cite{C.McQuillan/2013/Winding} to systematically discover canonical paths and thus prove rapid mixing of Markov chains.
Fahrbach and Randall \cite{DBLP:conf/approx/FahrbachR19} prove that the  Glauber dynamics mixes exponentially slowly in the entire ferroelectric phase and a part of the anti-ferroelectric phase with some boundary conditions.
Besides, Cai and Liu~\cite{Cai/Liu/21/FPTAS_Vertex} gave a fully polynomial-time approximation scheme (FPTAS) for the six-vertex and the eight-vertex models on square lattice graphs when the temperature is sufficiently low.

Since physicists discovered some remarkable algorithms on the planar graph, such as the FKT algorithm 
\cite{Temperley/Fisher/1961/Dimer,
Kasteleyn/1961/Dimer,
Kasteleyn/1963/Dimer,
Kasteleyn/1967/Graph_Crystal}, the planar structure has attracted considerable attention.
Cai et al. \cite{J.Cai/Z.Fu/S.Shao/2021/Planar_Six_Vertex} established a complexity trichotomy for the six-vertex model by taking planar tractability into account.
In terms of approximate complexity, Cai and Liu~\cite{Cai/Liu/2020/FPRAS_Eight_Vertex} gave the FPRAS for a subregion of the low temperature regime on planar graphs, using the holographic transformation method introduced by Valiant~\cite{Valiant/2004/Holographic_Algorithms}.
Their work demonstrated that the eight-vertex model is the first problem with the provable property of being NP-hard to approximate on general graphs (and even \#P-hard for planar graphs in exact complexity), while possessing FPRAS in significant regions of its parameter space on both bipartite and planar graphs.

In this paper, we study the \emph{four-vertex model}, a sub-model of the six-vertex model.
In this model we set $b=0$, which means that the local configurations 3 and 4 in \cref{six-type} are disallowed, once the local labeling on the four edges around each vertex is determined.
We note that the roles of $a$ and $b$ are symmetric in the six-vertex model, as rotating the edge labelings by 90 degrees on all vertices in a graph turns an input from $(a, 0, c)$ to one from $(0, a, c)$.
As a remark, sub-models of the six-vertex model with four permissible local configurations have been studied in physics~\cite{Bogoliubov/2008/Four_Vertex_Tilings, Bogoliubov/2019/Four_Vertex_Partition}.

Note that $(a,0,c)$ is in the low-temperature regime of the six-vertex model when $a \neq c$, hence is NP-hard to approximate on general 4-regular graphs~\cite{J.Cai/T.Liu/P.Lu/2019/Approximability_Six_Vertex, CaiLLY20/Eight}.
Although the problem is hard in the worst case, in this paper, we identify situations when an efficient approximation algorithm exist.

Our approach involves a technique called ``circuit decomposition'' which enables us to reduce the four-vertex model to the Ising model, a fundamental model that has been extensively studied.
We then identify a system of linear equations over GF($2$) that, when solved, reveals a method for designing a rapid mixing Markov chain (known as the \emph{worm process}) for counting and sampling the four-vertex model.
Notably, the constraint functions of the four-vertex model are \emph{unwindable}, making this the first attempt to design an FPRAS for the six-vertex model with unwindable constraint functions.

When it comes to planar graphs, the partition function of the four-vertex model can be computed exactly in polynomial time because the constraint functions under the parameter setup $(a, 0, c)$ are matchgate signatures~\cite{J.Cai/Z.Fu/S.Shao/2021/Planar_Six_Vertex}.
However, in order to (at least approximately) sample from the state space, the traditional sampling-via-counting encounters an obstacle since in general the six-vertex model is not known to be self-reducible~\cite{Jerrum/Valiant/Vazirani/1986/counting_sampling}.
In this paper, we propose a ``canonical labeling'' for the planar four-vertex model based on the structure of planar 4-regular graphs.
Our work shows that, once the input planar graph has canonical labeling, the aforementioned system of linear equations is always solvable, enabling us to use the worm process to give an efficient sampling algorithm. 
Moreover, under this canonical labeling, the planar four-vertex model can be reduced to the Ising model that is defined on the medial graph of its underlying graph.
As a remark, we feel that it is unlikely that a natural class of graphs would all satisfy the algebraic criteria — a linear system is solvable. 
This further emphasizes our discovery that the class of planar graphs, including the primitive case in statistical physics (the square lattice), are all found to satisfy this criteria.

This paper is organized as follows: 
In Section 2, we formalize the six-vertex model and give other basic definitions;
then, we state the four-vertex model and prove that its constraint function is generally unwindable.
In Section 3, we present the ``circuit decomposition'' technique and the reduction to the Ising model, followed by the study the approximability of the four-vertex model.
In Section 4, we analyze the worm process for the four-vertex model and prove it is rapid mixing.
In Section 5, we introduce the equivalence between the planar four-vertex model under canonical labeling and the Ising model, and provide a sampling algorithm for the planar four-vertex model.

\section{Preliminaries}

\subsection{Six-vertex model}

To ease technical discussion, we present the six-vertex model as a Holant problem.
The Holant problem is defined as follows.
A signature or constraint function of arity $k$ is a map $f:\{0,1\}^k \to \mathbb{C}$.
In this paper we restrict $f$ to take values in $\mathbb{R}^+$.
Let $\mathcal{F}$ be a set of signatures.
A signature grid $\Omega = ( G,\pi )$ is a tuple, where $G=( V,E )$ is a graph, $\pi$ labels each $v\in V$ with a signature $f_v\in\mathcal{F}$ and the incident edges $E (v)$ at $v$ are identified as input variables of $f_v$, also labeled by $\pi$.
Every configuration $\sigma \in \{0,1\}^E$  gives an evaluation $\prod_{v\in V}f_v(\sigma|_{E(v)})$ where $\sigma|_{E(v)}$ denotes the restriction of $\sigma$ to $E (v)$. 
The problem ${\rm Holant}(\mathcal{F})$ on an instance $\Omega$ is to compute
\begin{equation*}
    {\rm Holant}(\Omega ; \mathcal{F})\coloneqq\sum_{\sigma \in \{0,1\}^E}\prod_{v\in V}f_v(\sigma|_{E(v)}).
\end{equation*}
We denote ${\rm Holant}(\mathcal{F} | \mathcal{G})$ for Holant problems over signature grids with a bipartite graph $( U, V, E)$ where each vertex in $U$ (or $V$) is labeled by a function in $\mathcal{F}$ (or $\mathcal{G}$, respectively).

To write the six-vertex model on a 4-regular graph $ G=( V,E ) $ as a Holant problem, we consider its edge-vertex incidence graph.
The edge-vertex incidence graph $G^\prime=( U_E,U_V,E^\prime)$ is a bipartite graph where $(u_e,u_v)\in U_E\times U_V$ is an edge in $E^\prime$ iff $e\in E$ is incident to $v\in V$.
A configuration of the six-vertex model on $G$ is an edge 2-coloring on $G^\prime$, namely $\sigma: E^\prime\to\{0,1\}$.
We model an orientation of edges $e\in E$ by requiring ``one-0 one-1'' for the two edges incident to $u_e\in U_E$ and model the ice rule by requiring ``two-0 two-1'' for the four edges incident to each vertex $u_v\in U_V$.

The ``one-0 one-1'' requirement is a binary Diseqality constraint, denoted by ($\neq_2$). 
We write the values of a $4$-ary function $f$ as a matrix $$M(f)=
\left[ {\begin{array}{*{20}{c}}
    {f_{0000}}&{f_{0010}}&{f_{0001}}&{f_{0011}}\\
    {f_{0100}}&{f_{0110}}&{f_{0101}}&{f_{0111}}\\
    {f_{1000}}&{f_{1010}}&{f_{1001}}&{f_{1011}}\\
    {f_{1100}}&{f_{1110}}&{f_{1101}}&{f_{1111}}
\end{array}} \right]$$ called the constraint matrix of $f$.
The constraint function of the six-vertex model with arrow reversal symmetry has the form $M(f)=
\left[ {\begin{array}{*{20}{c}}
    0&0&0&a\\
    0&b&c&0\\
    0&c&b&0\\
    a&0&0&0
\end{array}} \right]$.
The partition function of the six-vertex model can be expressed as a Holant problem
\begin{equation*}
    Z_{6V}(G;a,b,c)={\rm Holant}(\neq_2 | f). 
\end{equation*}
Notice that the partition functions are multiplicative over all connected components, so we may assume $G$ is connected.

In this paper, we focus on the  four-vertex model whose constraint matrix has the form $M(f^*)=
\left[ {\begin{array}{*{20}{c}}
    0&0&0&a\\
    0&0&c&0\\
    0&c&0&0\\
    a&0&0&0
\end{array}} \right]$. The partition function of the four-vertex model is
\begin{equation*}
    Z_{4V}(G;a,c)=\sum_{\tau\in\mathcal{EO}(G)}a^{n_1+n_2}c^{n_5+n_6}.
\end{equation*}
The exact computation of $Z_{4V}(G;a,c)$ is \#P-hard according to \cite{J.Cai/Z.Fu/M.Xia/2018/Complexity_Six_Vertex}.
Moreover, there is no FPRAS for it in general according to \cite{J.Cai/T.Liu/P.Lu/2019/Approximability_Six_Vertex}.

Note that $Z_{4V}(G;a,c)=c^{n_1+n_2+n_5+n_6}\sum_{\tau\in\mathcal{EO}(G)}\left(\frac{a}{c}\right)^{n_1+n_2}$ and $n_1+n_2+n_5+n_6$ is the number of vertices under the four-vertex model setting.
We sometimes normalize the constraint function by ignoring the constant factor $c^{n_1+n_2+n_5+n_6}$ for convenience.
Let $\beta = a/c$, the normalized constraint function is  $M(f^*)=
\left[ {\begin{array}{*{20}{c}}
    0&0&0&\beta\\
    0&0&1&0\\
    0&1&0&0\\
    \beta&0&0&0
\end{array}} \right]$.
In the remainder of this paper, we always assume that the constraint function has been normalized.

\subsection{Ising model}

The Ising model is a basic model in statistical physics.
Since it is simple in expression and has wide application, it has been widely studied in physics, mathematics and computer science.
From the algorithmic perspective, given a graph $G=(V,E)$, the Ising model on the graph $G$ with parameters $\beta_e, e\in E$ is defined as follows.
For any $\sigma\in{\{0, 1\}}^V$, the probability of being in configuration $\sigma$ is
\begin{equation*}
    \pi(\sigma)=\frac{\prod_{e\in m(\sigma)}\beta_e}{Z_{Ising}},
\end{equation*}
where $m(\sigma)$ is the set of mono-chromatic edges in $\sigma$, i.e., $(u,v)\in m(\sigma)$ iff $\sigma(u)=\sigma(v)$; and its partition function is defined as
\begin{equation*}
    Z_{Ising}=\sum_{\sigma\in{\{0, 1\}}^V}\prod_{e\in m(\sigma)}\beta_e.
\end{equation*}

The computation of the partition function is \#P-complete.
For the approximate computation,
when $\beta_e>1$, for all $e\in E$, the system is ferromagnetic, and the partition function can be approximated in polynomial time via MCMC \cite{Jerrum/Sinclair/1993/Approximating_Ising};
when $\beta_e<1$, for some $e\in E$, the system is anti-ferromagnetic, and the partition function remains NP-hard to approximate.

\subsection{Approximation algorithm}\label{Approximation algorithm}

A randomized approximation scheme for a counting problem $f:\Sigma^*\to\mathbb{R}$ is a randomized algorithm that takes as input an instance $x\in\Sigma^*$ and an error tolerance $\varepsilon>0$,
and outputs a number $Y$ such that, for every instance $x$,
\begin{equation*}
    {\rm Pr}[(1-\varepsilon)f(x)\le Y\le(1+\varepsilon)f(x)]\ge\frac{3}{4}.
\end{equation*}
If the algorithm runs in time bounded by a polynomial of $|x|$ and $\varepsilon^{-1}$, we call it a \textit{fully polynomial randomized approximation scheme(FPRAS)}.

A standard approach to reducing almost uniform sampling to approximate counting involves using the powerful Markov chain Monte Carlo method. 
However, to ensure efficiency, it is necessary to bound the convergence rate of the Markov chains used in the sampler.
For a Markov chain with finite state space $\Omega$, transition matrix $P$ and stationary distribution $\pi$, its mixing time is defined as
\begin{equation*}
    \tau_\varepsilon(P)\coloneqq\min\{ t:\max_{x\in\Omega}{\left\lVert P^t(x,\cdot)-\pi\right\rVert }_{\rm TV}\le \varepsilon \},
\end{equation*}
where ${\left\lVert \cdot \right\rVert }_{\rm TV}$ is the total variation distance, i.e.,
\begin{equation*}
    {\left\lVert \pi-\pi^\prime \right\rVert }_{\rm TV}\coloneqq\frac{1}{2}\sum_{x\in\Omega}|\pi(x)-\pi^\prime(x)|.
\end{equation*}

One of the most useful techniques to bound mixing time is canonical path.
Let $\Psi \subseteq \Omega$ be a nonempty set and
\begin{equation*}
    \Gamma =\{\gamma_{xy}:(x,y)\in\Omega\times\Psi \}
\end{equation*}
be a collection of paths, where $\gamma_{xy}$ is a \textit{canonical path} from $x$ to $y$ using the transition of the Markov chain.
The ``local'' congestion $\varrho(\Psi;\Gamma)$ associated with the state set $\Psi$ and these paths is
\begin{equation*}
    \varrho(\Psi;\Gamma) \coloneqq \max_{(z,z^\prime)\in \Omega^2,P(z,z^\prime)>0}\frac{L(\Gamma)}{\pi(\Psi)\pi(z)P(z,z^\prime)}\sum_{(x,y)\in\Omega\times\Psi,\gamma_{xy}\ni (z,z^\prime)}\pi(x)\pi(y) ,
\end{equation*}
where $L(\Gamma)$ is the length of the longest path in $\Gamma$.
The congestion can be used to bound the mixing time by the following theorem.
\begin{theorem}\label{thm: cogistion and mixing time}
    {\upshape (\cite{Schweinsberg/2002/Relaxation_Time})}
    For an irreducible and aperiodic Markov chain with finite state space $\Omega$, transition matrix $P$, stationary distribution $\pi$ and any initial state $x_0\in\Omega$,
    \begin{equation*}
        \tau_\varepsilon(P)\le 4\varrho(\Psi;\Gamma)(\ln\frac{1}{\pi(x_0)}+\ln\frac{1}{\varepsilon}).
    \end{equation*}
\end{theorem}

\subsection{Windability}

Windability, proposed by McQuillan \cite{C.McQuillan/2013/Winding}, is a technique used to systematically design canonical paths. 
It reduces the task of designing canonical paths to solving a set of linear equations, making it easier to design them. 
Specifically, for a Holant problem, windability allows for automatic discovery of canonical paths if all constraint functions are ``windable''. 
Additionally, Huang et al. \cite{L.Huang/P.Lu/C.Zhang/2016/Canonical_Path} simplified the conditions for checking whether a function is windable. 
Since then, the application of windability has been further extended. 
The definition of windability is as follows:
\begin{definition}\label{def: windability}
	For any finite set $J$ and any configuration $x \in \{0,1\}^J$, define $\mathcal{M}_x$ to be the set of partitions of $\{i\mid x_i=1 \}$ into pairs and at most one singleton. A function $f:\{0,1\}^J \to \mathbb{Q}^+$ is windable if there exist values $B(x,y,M) \ge 0$ for all $x,y \in \{0,1\} ^J$ and all $M \in \mathcal{M}_{x \oplus y}$ satisfying:
	\begin{itemize}
	\item $f(x)f(y) = \sum_{M \in \mathcal{M}_{x \oplus y}} B(x,y,M) $  for all $x,y \in \{ 0,1\} ^J$.
	\item $B(x,y,M) = B(x \oplus S,y \oplus S,M)$ for all $x,y \in \{ 0,1\} ^J$ and $S \in M \in \mathcal{M}_{x \oplus y}$.
	\end{itemize}
	Here $x \oplus S$ denotes the vector obtained by changing ${x_i}$ to $1 - {x_i}$ for the one or two elements i in S.
\end{definition}

Note that for the six-vertex model with arrow reversal symmetry, the constraint function is windable if $a^2\le b^2+c^2$ and $b^2\le c^2+a^2$ and $c^2\le a^2+b^2$, i.e., \cite{J.Cai/T.Liu/P.Lu/2019/Approximability_Six_Vertex} gave an FPRAS for the six-vertex model when the constraint function is windable. 
We prove that the constraint function $f^*$ is unwindable in the following proposition, i.e., we will give an FPRAS for the six-vertex model without windability.

\begin{proposition}\label{unwindable}
	The constraint function ${f^*}$  with $a\neq c$ is unwindable.
\end{proposition}

\begin{proof}
    Assume to the contrary that the constraint function $f^*$ is windable.
	First, we set $x=0101$ and $y=1010$, then 
	\begin{equation*}
		\mathcal{M}_{x \oplus y}=\left\{M_1,M_2,M_3\right\},
	\end{equation*}
	where $M_1=\{(x_1,x_2),(x_3,x_4)\}$, $M_2=\{(x_1,x_3),(x_2,x_4)\}$, $M_3=\{(x_1,x_4),(x_2,x_3)\}$. 
	According to \cref{def: windability}, there exists $B\left( {x,y,M} \right) \ge 0$ such that ${f^*}(x) \cdot {f^*}(y) = \sum\nolimits_{M \in {{{\cal M}'}_{x \oplus y}}} {B(x,y,M)}  = {c^2}$.
	Moreover, for $M_1$, we have 
	\begin{align*}
		B\left(x,y,{M_1} \right)= B(x\oplus\{x_1,x_2\},y\oplus\{x_1,x_2\},M_1)= B\left( {1001,0110,{M_1}} \right) = 0
	\end{align*}
	since ${f^*}(1001) \cdot {f^*}(0110) = 0$. 
	Similarly, for $M_3$, we have
    \begin{align*}
		B\left(0101,1010,{M_2}\right) = B\left( {1111,0000,{M_2}} \right) = 0.
	\end{align*}
    Thus,
    \begin{equation*}
        c^2=B(0101,1010,M_3).
    \end{equation*}
    Now we set $x=0011$ and $y=1100$.
    By a very similar analysis for $x=0101$ and $y=1010$, we get 
    \begin{equation*}
        a^2=B(0011,1100,M_3).
    \end{equation*}
	Since $B(0101,1010,M_3)=B(0011,1100,M_3)$, this implies that $a=c$ which contradicts that $a\neq c$.
	Therefore, $f^*$ with $a\neq c$ is unwindable. 
\end{proof}

\section{Reduction}

\subsection{Circuit decomposition}

Given a simple graph $G=(V, E)$ and its corresponding edge-vertex graph $G^\prime=( U_E,U_V,E^\prime)$, 
recall that we label $f^*$ on each vertex in $U_V$ and $\neq_2$ on each vertex in $U_E$ in the Holant problems.
We divide the four variables of $f^*$ into two pairs $(x_1, x_4)$ and $(x_2, x_3)$, and note that the constraint function $f^*$ forces the variables in the same pair to take opposite values in each valid configuration, 
such that corresponding edges must be pairwise incoming and outgoing. 
This observation motivates the following \emph{circuit decomposition} for the underlying graph:

\begin{itemize}
	\item Select an arbitrary vertex and an incident edge, and begin tracing a trail from this edge;
	\item for the vertices in $U_E$ with degree $2$, the trail proceeds directly through them;
	\item for the vertices in $U_V$ with degree $4$, the trail proceeds through the edges corresponding to variables in the same pair. Specifically, if a path enters a vertex from the edge $x_1$ or $x_4$ ($x_2$ or $x_3$), then it leaves the vertex to the edge $x_4$ or $x_1$ ($x_3$ or $x_2$), respectively;
	\item when the trail returns to the starting vertex, it forms a circuit. Remove this circuit from the graph and repeat the process until the graph is empty.
\end{itemize}
See \cref{fig:label} for an example.

\begin{figure}[ht]		
	\centering
	\includegraphics[width=6cm]{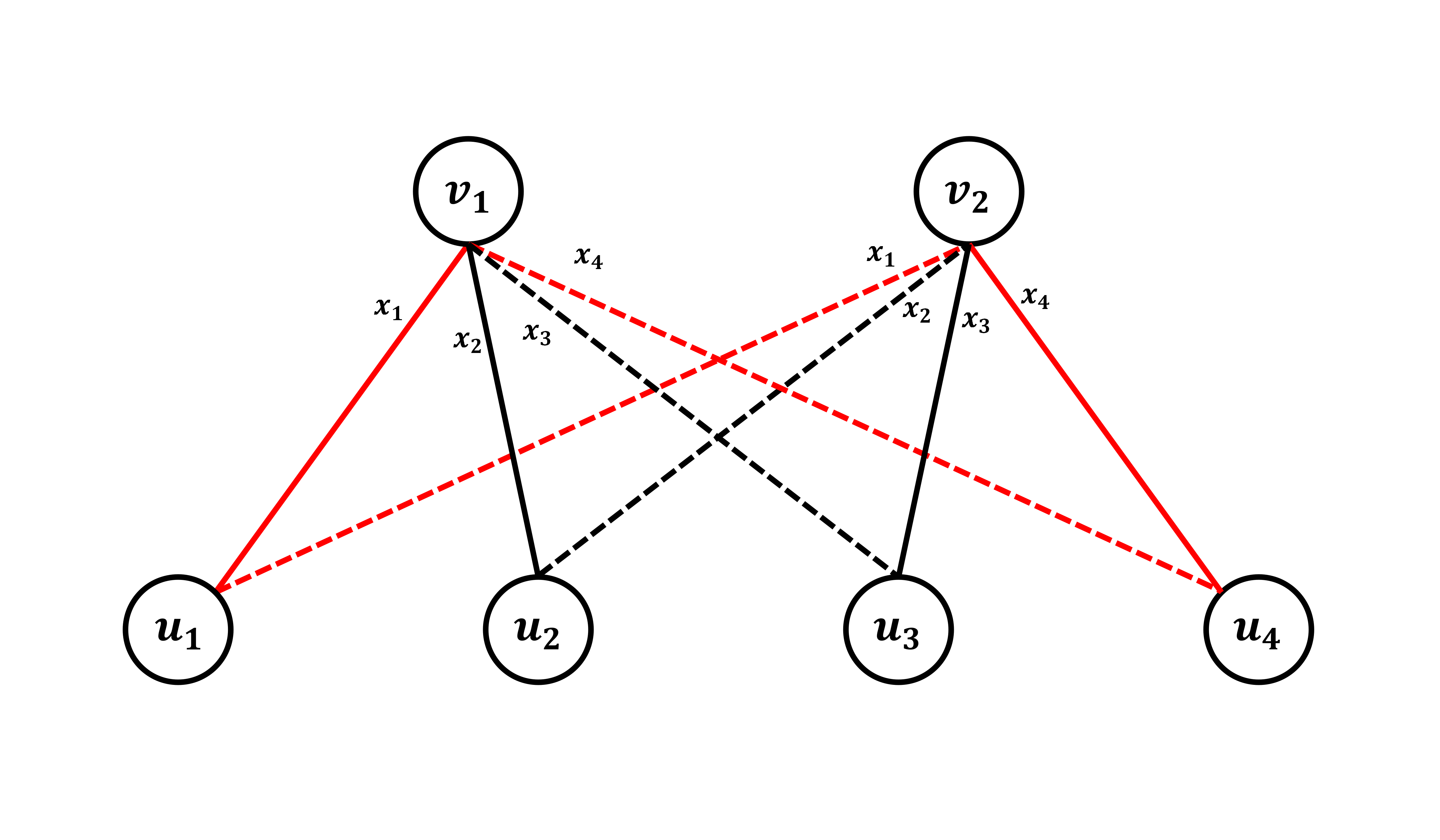}
	\caption{An example of circuit decomposition: the red edges and the black edges form two circuits 
        $C_1(v_1\to u_1\to v_2\to u_4\to v_1)$ and $C_2(v_2\to u_2\to v_1\to u_3\to v_2)$.
		In the same circuit, the dashed edges and the solid edges take opposite assignments.
	} 
	\label{fig:label}	
\end{figure}

Let $\mathcal{C}$ denote the set of circuits resulting from the circuit decomposition, and let $m = |\mathcal{C}|$.
We index each circuit in $\mathcal{C}$ and denote the set of circuits as $\mathcal{C}=\{C_1,C_2,\cdots,C_m\}$.
Now we redefine the valid configuration of the four-vertex model in terms of the circuit decomposition $\mathcal{C}$.
Note that the constraint function $f^*$ and $\neq_2$ force the edges in the same circuit in $\mathcal{C}$ take the values  $\{0,1\}$ alternately.
For each circuit ${C_i} $, we arbitrarily select an initial edge  $e_i$. 
The assignments of all edges in $C_i$ are determined by the assignment of $e_i$ in a valid configuration.
Therefore, we can define the assignment of the circuit $C_i$ as the value of the initial edge $e_i$, i.e., a valid configuration assigns a value for each circuit in $\mathcal{C}$. 
Conversely, given an assignment for each $C_i\in\mathcal{C}$, it gives a valid configuration for the four-vertex model.

We remark that a similar (but not identical) idea of circuit decomposition was used in \cite{Cai/Liu/20/8V_PM} to establish connections between the eight-vertex model (after holographic transformations) and the Ising model.

\subsection{Vertex classification}

Note that if four edges of a vertex in $U_V$ belong to only one circuit, the value of the constraint function at this vertex remains constant regardless of the circuit's value.
Therefore, we can assume that each vertex in $U_V$ belongs to two different circuits.
To classify the common vertices of two circuits, we differentiate them based on the variation in constraint function values when they have the same input.
Formally, consider circuit $C_i$ and $C_j$ have a common vertex $v$, in the sense of symmetry, the constraint function at $v$ might have four cases as follows:
\begin{equation*}
\begin{split}
    f^*_v(x_1,x_2,x_3,x_4)=f^*(C_i,C_j,\overline{C_j},\overline{C_i}); f^*_v(x_1,x_2,x_3,x_4)=f^*(\overline{C_i},\overline{C_j},C_j,C_i);\\
    f^*_v(x_1,x_2,x_3,x_4)=f^*(C_i,\overline{C_j},C_j,\overline{C_i}); f^*_v(x_1,x_2,x_3,x_4)=f^*(\overline{C_i},C_j,\overline{C_j},C_i).
\end{split}
\end{equation*}

The first two cases occur when the constraint function at vertex $v$ takes the value $\beta$ if circuits $C_i$ and $C_j$ are assigned the same value and $1$ if they are assigned different values. 
The last two cases are the opposite. 
If the constraint function at $v$ falls into the first two cases, we call $v$ an \textit{agree-vertex} between circuits $C_i$ and $C_j$. 
Otherwise, we call $v$ a \textit{disagree-vertex} between $C_i$ and $C_j$.
For two circuits $C_i,C_j\in\mathcal{C}$, let
\begin{align*}
    &A(i,j)\coloneqq\# \text{agree-vertices between } C_i\text{ and } C_j,\\
    &D(i,j)\coloneqq\# \text{disagree-vertices between } C_i\text{ and } C_j.
\end{align*}

Now we define the graph of circuit $G_C=(\mathcal{C},E_C)$.
The vertex set of $G_C$ is $\mathcal{C}$.
Edge $(u,v)\in E_C$ iff circuit $C_u$ and $C_v$ have at least one common vertex.
Based the fact that we have stated in this section, the four-vertex model on $G$ can be reduced to a spin system with local constraint function on $G_C$.
 
The configuration of the spin system is one of the $2^m$ possible assignments $\sigma :\mathcal{C}\to \{0,1\}$ of states to vertices.
Label the local constraint function $f^C_{(u,v)}$ on each $(u,v)\in E_C$ which has the constraint matrix $M(f^C_{(u,v)})=
\left[ {\begin{array}{*{20}{c}}
    \beta^{A(u,v)}&\beta^{D(u,v)}\\
    \beta^{D(u,v)}&\beta^{A(u,v)}
\end{array}} \right]$.
The partition function of the four-vertex model can be written as
\begin{align*}
    Z_{4V}
    &=\sum_{\sigma :\mathcal{C}\to \{0,1\}}\prod_{\substack{(u,v)\in E_C\\\sigma(u)=\sigma(v)}}\beta^{A(u,v)}\prod_{\substack{(u,v)\in E_C\\\sigma(u)\neq\sigma(v)}}\beta^{D(u,v)}\\
    &=\prod_{(u,v)\in E_C}\beta^{D(u,v)}\sum_{\sigma :\mathcal{C}\to \{0,1\}}\prod_{\substack{(u,v)\in E_C\\\sigma(u)=\sigma(v)}}\beta^{A(u,v)-D(u,v)}.
\end{align*}

Ignoring the factor $\prod_{(u,v)\in E_C}\beta^{D(u,v)}$ which can be computed directly, our goal is to compute the partition function
\begin{equation}\label{eq:4v_as_ising}
Z_{4V}=\sum\limits_{\sigma :\mathcal{C}\to \{0,1\}}\prod\limits_{\substack{(u,v)\in E_C\\\sigma(u)=\sigma(v)}}\beta^{A(u,v)-D(u,v)}.
\end{equation}
It is evident that this conversion transforms the four-vertex model on $G$ to the Ising model on $G_C$. 
However, since the relative magnitudes of $A(u,v)$ and $D(u,v)$ are uncontrollable, the interaction between any two vertices in $G_C$ could be either ferromagnetic or antiferromagnetic, making the problem hard to tackle.
Indeed, as previously noted, $Z_{4V}$ is generally NP-hard to approximate on $4$-regular graphs.

\begin{theorem}[\cite{CaiLLY20/Eight}]
There can be no FPRAS for the four-vertex model with $a\neq c$ unless RP = NP.
\end{theorem}

Despite the problem is hard in the worst case, we can identify situations in which an FPRAS exists. 
Specifically, an FPRAS is available if the graph of circuit $G_C$ can be transformed into an instance with only ferro-Ising types of vertex interactions.

Although we assumed that the initial edges of circuits were fixed in the above analysis, we can resize $A(u,v)$ and $D(u,v)$ by changing the initial edges of some circuits.
Specifically, if we change exactly one circuit's initial edge to an edge adjacent to it for any $u, v \in \mathcal{C}$, the agree-vertices (disagree-vertices) between $u$ and $v$ will become disagree-vertices (agree-vertices), respectively.
If we can change some initial edges so that $A(u,v) \ge D(u,v)$ for all $u, v$ when $\beta > 1$, \cref{eq:4v_as_ising} becomes a ferro-Ising type of computation.
The case for $\beta < 1$ is similar.
We then express the conditions as systems of linear equations over GF($2$).

Let $X_i$ indicate whether $C_i$ changes its initial edge, i.e., $X_i=1$ meaning $C_i$ change its initial edge and $X_i=0$ meaning $C_i$ does not change its initial edge.
For $\beta>1$, we can write down the following system of linear equations over GF($2$).
\begin{align}\label{equ: beta>1 equations}
    X_u\oplus X_v={\bf I}(A(u,v)<D(u,v)), \quad \forall (u,v)\in E_C,
\end{align}
where $\oplus$ is the XOR operator and ${\bf I}(\cdot)$ is the indicator function.
Similarly for $\beta<1$, we can write down the following system of linear equations.
\begin{align}\label{equ: beta<1 equations}
    X_u\oplus X_v={\bf I}(A(u,v)>D(u,v)), \quad \forall (u,v)\in E_C.
\end{align}

It is worth noting that both systems are relatively sparse, with $E_c$ equations but each equation only having two variables. 
As a result, we can solve them in at most $O(m^2|E_C|)$ time by using Gaussian elimination, which is polynomial.
In \cref{worm process}, we prove the following theorem by analyze the worm process.
\begin{theorem}\label{four-vertex FPRAS}
    If \cref{equ: beta>1 equations} or \cref{equ: beta<1 equations}  has a solution, there is an FPRAS for the $Z_{4V}(G;\beta,1)$ with $\beta>1$ or $\beta<1$ respectively.
\end{theorem}

\section{Worm process for the four-vertex model}\label{worm process}

In this section, we analyze the worm process for the four-vertex model with $\beta>1$ and \cref{equ: beta>1 equations} having a solution.
This also applies to the analysis of $\beta<1$ and \cref{equ: beta<1 equations} having a solution.
It gives an FPRAS for the partition function.
The worm process, introduced by Prikof'ev and Svistunov \cite{Prokof/Boris/2001/Worm_Algorithm}, is a Markov chain that transitions between even subgraphs and near-even subgraphs.
The worm process of ferromagnetic Ising model has been proven to mix rapidly \cite{Collevecchio/Garoni/Hyndman/Tokarev/2016/Worm_Raipd_Mixing}.

After changing circuits' initial edges according to the solution, let $\beta_e=\beta^{A(u,v)-D(u,v)}$ for all $e=(u,v)\in E_C$.
We can rewrite the partition function of the four-vertex model as
\begin{equation}\label{four vertex Ising}
    Z_{4V}=\sum_{\sigma :\mathcal{C}\to \{0,1\}}\prod_{(u,v)\in E_C}\beta_e ^{{\bf I}(\sigma(u)=\sigma(v))}.
\end{equation}

Let $\Omega_k$ be the set of subgraphs of $G_C$ where exactly $k$ many vertices have odd degrees.
The state space of the worm process is $\Omega_{worm}\coloneqq\Omega_0\cup\Omega_2$.
There is a famous equivalence between \cref{four vertex Ising} and the partition function of the even subgraph model which can be explained via a holographic transformation (e.g., see \cite{Guo/Jerrum/2017/Cluster_Raipd_Mixing}).
That is
\begin{equation*}
    Z_{4V} = 2^{|\mathcal{B}|}\prod_{e\in E_C}\frac{\beta_e +1}{2}\sum_{S\in \Omega_0}\prod_{e\in S}\frac{\beta_e - 1}{\beta_e + 1}.
\end{equation*}
Let $x_e\coloneqq(\beta_e - 1)/(\beta_e + 1)$. 
For any $S\subseteq E_C$, let $w(S)\coloneqq\prod_{e\in S}x_e$ and $Z_k\coloneqq\sum_{S\in\Omega_k}w(S)$.
In the view of holographic transformation, it is known that $Z_k\le\binom{m}{k}Z_0$ (again, see \cite{Guo/Jerrum/2017/Cluster_Raipd_Mixing}).
The weight of a subset $S\subseteq E_C$ of the worm process is defined as $w_{worm}(S)=\xi (S)w(S)$ where
\begin{eqnarray*}
    \xi(S)\coloneqq 
    \begin{cases}
    m &\text{if }S\in\Omega_0\\
    2 &\text{if }S\in\Omega_2\\
    0 &\text{otherwise}
    \end{cases}.
\end{eqnarray*}
Let $Z_{worm}\coloneqq\sum_{S\in\Omega_{worm}}w_{worm}(S)=m Z_0+2Z_2$.
The worm measure is defined as $\pi_{worm}(S)\coloneqq \frac{w_{worm}(S)}{Z_{worm}}$.
To apply \cref{thm: cogistion and mixing time}, we use the following lemma to bound the worm measure.

\begin{lemma}\label{lem: bound worm measure}
    For all $S\in\Omega_{worm}$, we have
    \begin{equation*}
        w_{worm}(S)\geq \frac{1}{2}\left(\frac{x_{min}}{2}\right)^{|E_C|}
    \end{equation*}
    where $x_{\min}=\min_{e\in E_C}\{x_e\}$.
\end{lemma}

\begin{proof}
    By the definition of $Z_{worm}$ and $Z_k\le\binom{m}{k}Z_0$, it follows that
    \begin{align*}
		Z_{worm}
        &= m Z_0\left(1+\frac{2}{m}\frac{Z_2}{Z_0}\right) \\
		&\le m^2 Z_0\\
        &\le m^2 |\Omega_0|\\
        &= m^2 2^{|E_C|-m+1}.
	\end{align*}
    In the last equation, we have used the fact that any finite connected graph with $m$ vertices and $|E_C|$ edges has $2^{|E_C|-m+1}$ even subgraphs.
    
    Since $w_{worm}(S)\geq 2w(S)\geq 2x_{\min}^{|E_C|}$ holds for any $S\in\Omega_{worm}$, we have
    \begin{align*}
		\pi_{worm}(S)
		&\geq \frac{2x_{\min}^{|E_C|}}{m^2 2^{|E_C|-m+1}}\\
        &=\frac{1}{2}\left(\frac{x_{\min}}{2}\right)^{|E_C|}\frac{2^{m+1}}{m^2}\\
        &\geq \frac{1}{2}\left(\frac{x_{\min}}{2}\right)^{|E_C|}.
	\end{align*}
\end{proof}

Now suppose that the current state is $A$, the transition strategy of the worm process is described as follows:
\begin{itemize}
    \item If $A\in \Omega_0$:
        \begin{enumerate}
            \item choose a vertex $v\in \mathcal{C}$ uniformly at random,
            \item choose a neighbor $v \sim u$ uniformly at random,
            \item propose $A\mapsto A\oplus (u,v)$.
        \end{enumerate}
    \item If $A\in \Omega_2$:
        \begin{enumerate}
            \item choose an odd vertex $v\in \partial A$ uniformly at random,
            \item choose a neighbor $v \sim u$ uniformly at random,
            \item propose $A\mapsto A\oplus (u,v)$.
        \end{enumerate}
\end{itemize}
Here $\partial A$ denotes the set of odd vertices of $A$.
It is easily to verity that $A\oplus uv$ always in $\Omega_{worm}$.
So these proposals are well-defined.
To make the transition matrix be symmetric, we choose an appropriate Metropolis acceptance rate to modify the transition probability.
Moreover, to ensure the eigenvalues of the transition matrix are strictly positive, we consider the lazy version of the chain, i.e., stay the current state w.p. $\frac{1}{2}$ and do the proposal we described w.p. $\frac{1}{2}$ at each step.
The resulting transition matrix is
\begin{eqnarray*}
    P(A,A\oplus uv)\coloneqq
    \begin{cases}
    x_{(u,v)}^{{\bf I}((u,v)\notin A)}\frac{1}{2m}(\frac{1}{d(u)}+\frac{1}{d(v)}) &\text{if }A\in\Omega_0\\
    x_{(u,v)}^{{\bf I}((u,v)\notin A)}\frac{1}{4}(\frac{1}{d(u)}+\frac{1}{d(v)}) &\text{if }A\oplus (u,v)\in\Omega_0\\
    \min\left(1,\frac{d(u)}{d(v)}x_{(u,v)}^{{\bf I}((u,v)\notin A)-{\bf I}((u,v)\in A)}\right) \frac{1}{4d(u)} &\text{if }A,A\oplus (u,v)\in\Omega_2, u\in\partial A
    \end{cases}
\end{eqnarray*}
and other non-diagonal entries of $P$ are $0$, diagonal entries equal to $1$ minus other entries in the same row.
The congestion of the worm process with above transition matrix is bounded by the following lemma, which is a standard canonical path argument.
\begin{lemma}\label{lem: bound cogistion}
    There exists a choice of paths $\Gamma = \{\gamma_{x,y}:(x,y)\in\Omega_{worm}\times\Omega_0\}$, such that
    \begin{equation*}
        \varrho(\Omega_0;\Gamma)\le m ^5 |E_C|.
    \end{equation*}
\end{lemma}
\begin{proof}
    Let $I\in\Omega_{worm},F\in\Omega_0$ be two configurations, denoting the initial and final states respectively.
    Then $I\oplus F\in\Omega_{worm}$ and $\partial (I\oplus F)=\partial I$.
    First, we suppose $I\in\Omega_2$ with $\partial I=\{u,v\}$.
    Fix one of the shortest path $B_0$ between $u,v$, then $I\oplus F\setminus B_0\in\Omega_0$.
    Thus, we can decompose $I\oplus F\setminus B_0$ into a set of edge-disjoint cycles.
    Order the resulting cycles and denote them by $B_1,B_2,\cdots,B_r$.
    Moreover, specify a distinguishable initial vertex for each $B_i,i\ge 0 $, and a direction for each $B_i,i\ge 1$.
    Let $\{e_1,e_2,\cdots,e_k\}$ be the edges of $\{B_0,B_1,\cdots,B_r\}$ taken in ordered.
    Now we have $I\oplus F = \cup_{i=0}^{r}B_i=\{e_1,e_2,\cdots,e_k\}$.
    If $I\in\Omega_0$, just let $B_0=\varnothing$.

    The canonical path $\gamma_{IF}$ from $I$ to $F$ is defined to be $Z_0=I, Z_i=Z_{i-1}\oplus e_i$ and $Z_k=F$.
    Let $\Gamma_{worm}=\{\gamma_{IF}:(I,F)\in\Omega_{worm}\times\Omega_0\}$ be the set of such canonical paths.
    It is clearly that $\Gamma_{worm}\le |E_C|$ since every edge in $|E_C|$ can be used at most once.
    
    For each transition $(A,A^\prime)$, we use a combinatorial encoding for all paths through $(A,A^\prime)$.
    For $I\in\Omega_{worm},F\in\Omega_0$, let $\eta(I,F)=I\oplus F\oplus(A\cup A^\prime) $.
    We claim that $\eta:\Omega_{worm}\times\Omega_0\to\Omega_{worm}\cup\Omega_4$ is an injection.
    Given a transition $(A,A^\prime)$ and $U=\eta(I,F)$, all edges not in $(A\cup A^\prime)\oplus U$ have the same state in $I$ and $F$ since $(A\cup A^\prime)\oplus U=I\oplus F$.
    Let $e=A\oplus A^\prime$.
    According to the construction of the canonical paths, there is an ordering for all edges in $I\oplus F$.
    For each edge before $e$, its state in $A\cup A^\prime$ has been changed to that in $F$, and in $U$ is still the same as that in $I$.
    For each edge after $e$, its state in $A\cup A^\prime$ has been changed to that in $I$, and in $U$ is still the same as that in $F$.
    Thus, we can recover the unique $(I,F)$ from $(A,A^\prime)$ and $U$.
    Moreover, if $I\in\Omega_0$, then $\eta(I,F)\in\Omega_{worm}$; if $I\in\Omega_2$, then $\eta(I,F)\in\Omega_{worm}\cup\Omega_4$.

    Since $I\cap F\subseteq A\cup A^\prime\subseteq I\cup F$ and $U=I\oplus F\oplus(A\cup A^\prime)$, it follows that $U\cap(A\cup A^\prime)=I\cap F$ and $U\cup(A\cup A^\prime)=I\cup F$.
    We have
    \begin{align*}
		\frac{1}{\pi_{worm}(\Omega_0)}\frac{\pi_{worm}(I)\pi_{worm}(F)}{\pi_{worm}(A)P(A,A^\prime)}
        &=\frac{1}{m Z_0} \frac{m\xi (I)w(I)w(F)}{\xi(A)w(A)P(A,A^\prime)}\\
		&\le\frac{2m}{Z_0}\frac{\xi(I)w(I)w(F)}{w(A\cup A^\prime)}\\
		&=\frac{2m}{Z_0}\xi(I)w(\eta(I,F)).
	\end{align*}
    The inequality is a consequence of the definition of the transition probability and $d(u)\le m$ for each $u\in \mathcal{C}$.
    Now let $(A,A^\prime)$ be a maximally congested transition.
    We have
    \begin{align*}
		\varrho(\Omega_0;\Gamma_{worm})
        &\le \frac{2m |E_C|}{Z_0}\sum_{(I,F)\in\Omega_{worm}\times\Omega_0,\gamma_{IF}\ni (A,A^\prime)}\xi(I)w(\eta(I,F))\\
		&\le \frac{2m |E_C|}{Z_0}[m(Z_0+Z_2)+2(Z_0+Z_2+Z_4)]\\
        &\le 2m |E_C| \left[(m+2)+(m+2)\binom{m}{2}+2\binom{m}{4} \right]\\
        &\le m ^5 |E_C|.
	\end{align*}
    The second inequality is because that $\eta$ is an injection.
\end{proof}

Combining \cref{thm: cogistion and mixing time}, \cref{lem: bound worm measure} and \cref{lem: bound cogistion}, we immediately have the following theorem.
Furthermore, as we mentioned in \cref{Approximation algorithm}, this theorem also gives a proof to \cref{four-vertex FPRAS}.
\begin{theorem}
    For the four-vertex model $Z_{4v}(G;\beta,1)$, when $\beta>1$ and \cref{equ: beta>1 equations} has a solution or $\beta<1$ and \cref{equ: beta<1 equations} has a solution, the mixing time of the worm process $\tau_\varepsilon(P)\le 4n^5|E_C|^2(\frac{\ln 2\varepsilon^{-1}}{|E_C|}+\ln\frac{2}{x_{\min}})$.
\end{theorem}

\section{On planar graphs}

According to \cite{J.Cai/Z.Fu/S.Shao/2021/Planar_Six_Vertex}, the constraint functions of the four-vertex model under the parameter setup $(a, 0, c)$ are matchgate signatures, and thus the exact computation of its partition function can be done in polynomial time on planar graphs.
However, in order to (at least approximately) sample from the state space, the traditional sampling-via-counting encounters an obstacle since in general the six-vertex model is not known to be self-reducible~\cite{Jerrum/Valiant/Vazirani/1986/counting_sampling}.

In this section, we focus on the sampling complexity of the four-vertex model.
We identify a canonical way of labeling local edges around each vertex in the planar four-vertex model, such that \cref{equ: beta>1 equations} can always be satisfied, and thus the worm process studied in \cref{worm process} yields a sampling algorithm for $\beta > 1$.
It is unlikely that a natural class of graphs would satisfy the algebraic criteria of having a solvable linear system.
Our discovery that the class of planar graphs, which includes the fundamental case in statistical physics (the square lattice), satisfies this criteria further emphasizes its significance.
\subsection{Canonical labeling}

Consider the six-vertex model on the planar graph. 
It is well known that the dual graph of any 4-regular plane graph $G$ is bipartite \cite{Welsh/1969/Matroids}.
Thus, the faces of $G$ have a 2-coloring, i.e., a way to color the faces by using two colors (black and white) such that any two adjacent faces have different colors.
Without loss of generality, we assume that the outer face of $G$ is colored white. 
See \cref{face 2-coloring} for an example.
\begin{figure}[htbp]
    \centering
    \includegraphics[width=6cm]{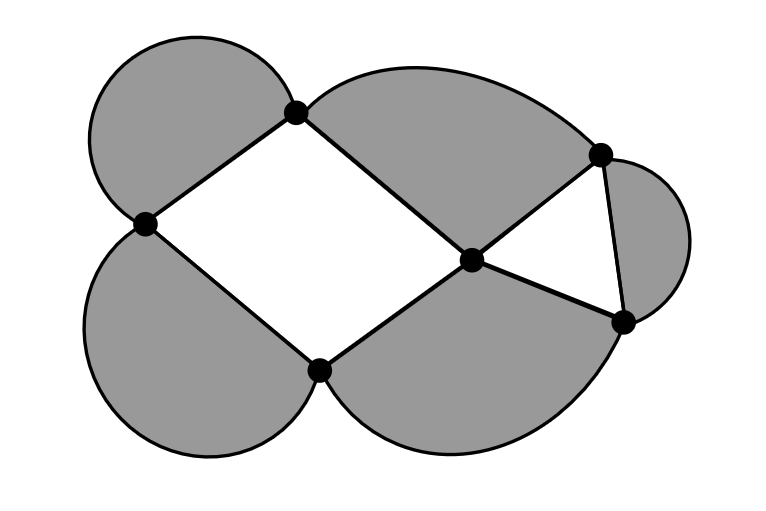}
    \caption{An example of face 2-coloring where the outer face is white.}
    \label{face 2-coloring}
\end{figure}

Based on the above facts, there is a canonical way of local edge labelings for the six-vertex model on a plane graph which divides the local orientations into three types:
\begin{itemize}
\item
The edges that clamp each black face are one-in-one-out, and the directions of the arrow-flows along two black faces are different. These vertices (1, 2 in \cref{planar six-type}) have weight $a$.  
\item
The edges which clamp each black face are two-in or two-out. These vertices (3, 4 in \cref{planar six-type}) have weight $b$.
\item
The edges which clamp each black face are one-in-one-out, and the directions of the arrow-flows along two black faces are the same. These vertices (5, 6 in \cref{planar six-type}) have weight $c$.
\end{itemize}
\begin{figure}[H]
    \begin{subfigure}[t]{0.15\textwidth}
        \centering
        \includegraphics[width=1\textwidth]{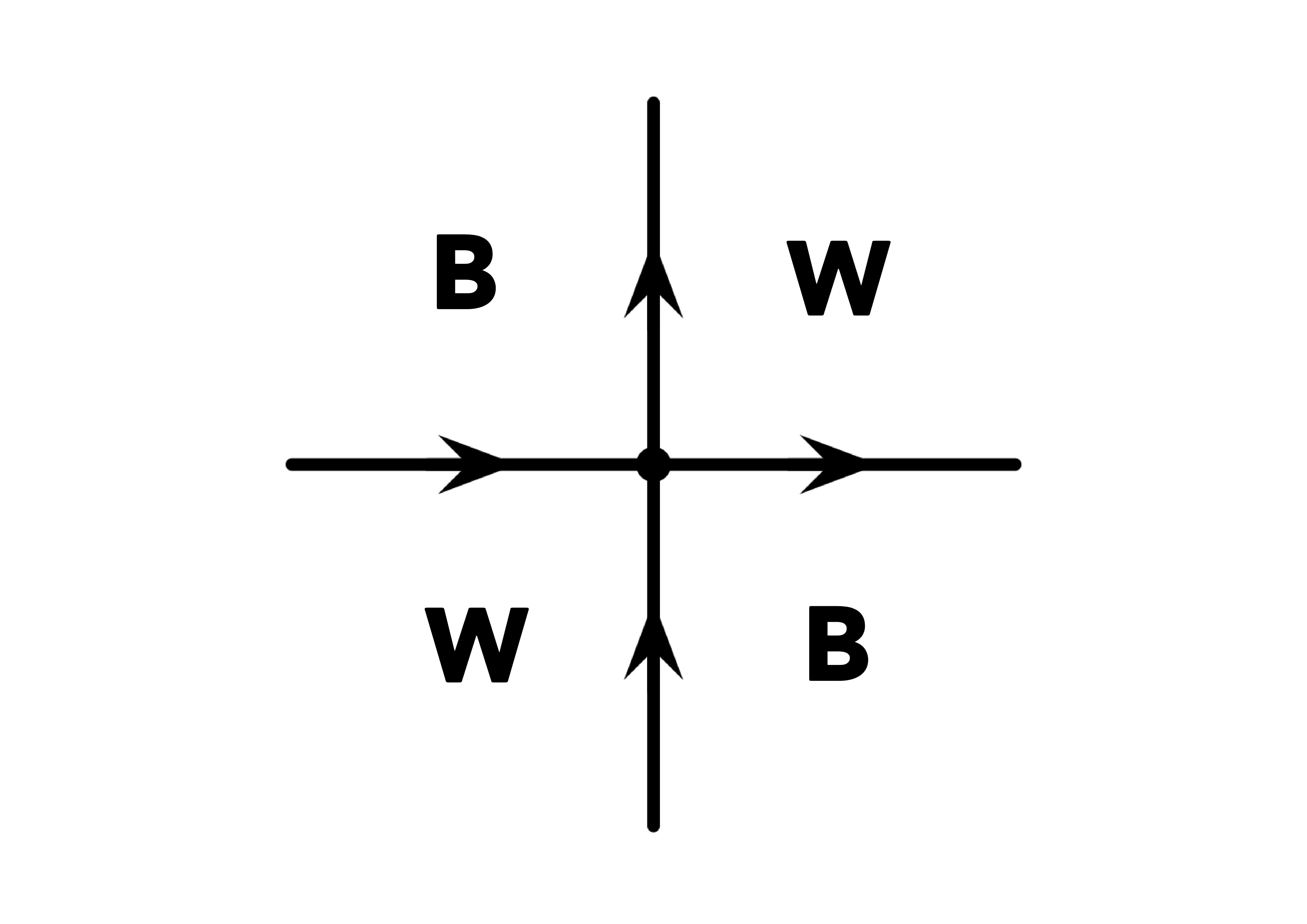}
        \subcaption*{1}
    \end{subfigure}
    \begin{subfigure}[t]{0.15\textwidth}
        \centering
        \includegraphics[width=1\textwidth]{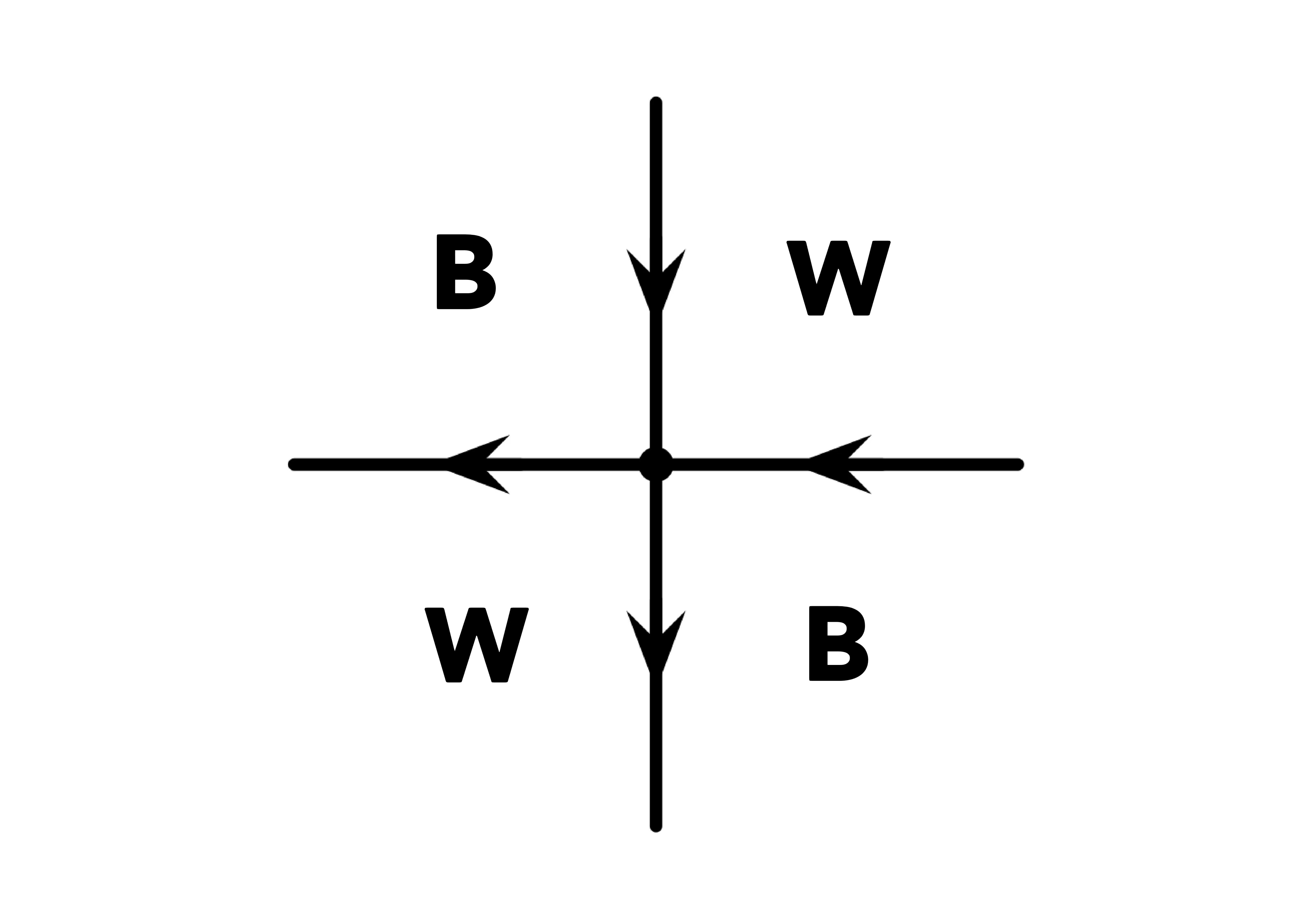}
        \subcaption*{2}
    \end{subfigure}
    \begin{subfigure}[t]{0.15\textwidth}
        \centering
        \includegraphics[width=1\textwidth]{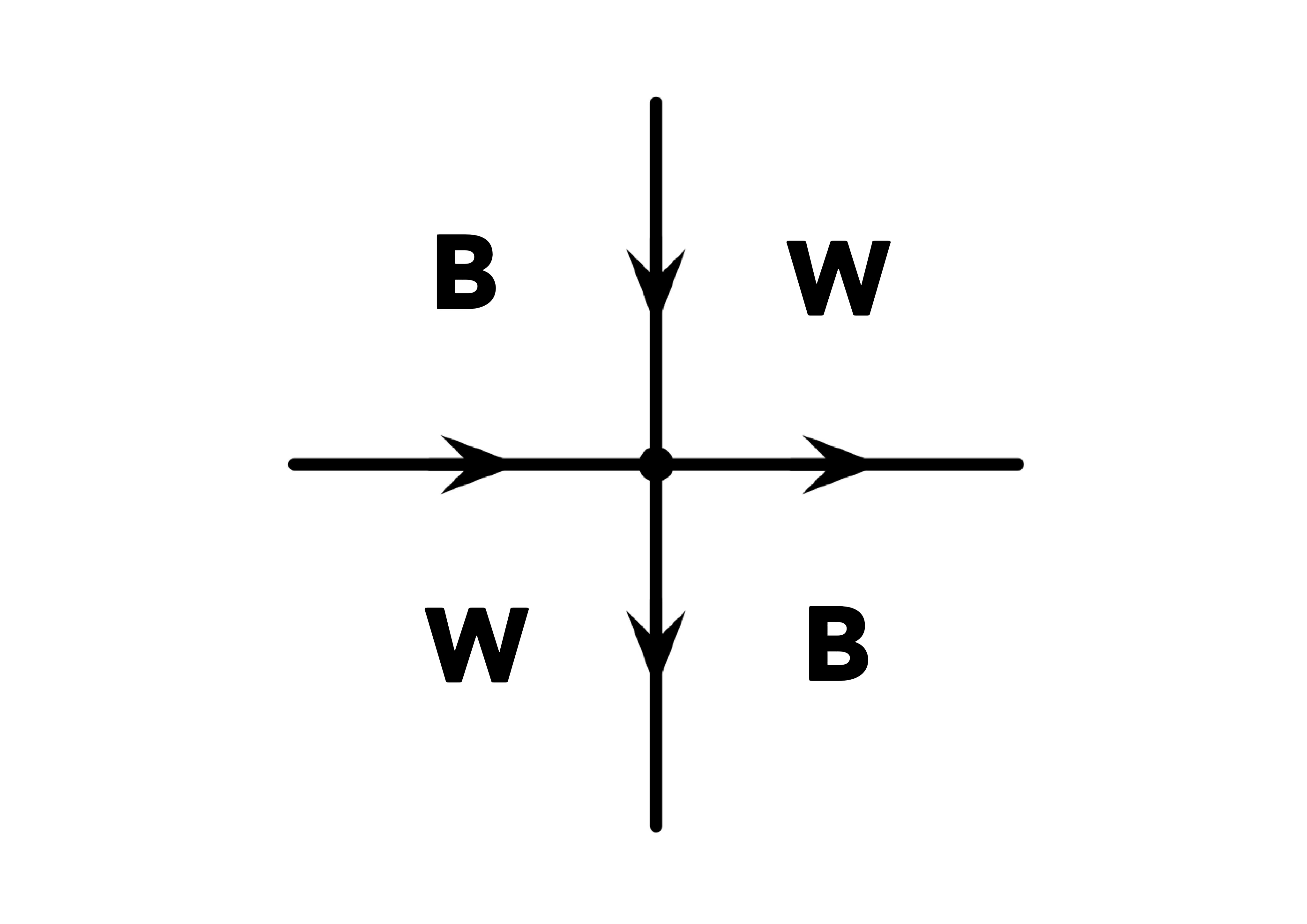}
        \subcaption*{3}
    \end{subfigure}
    \begin{subfigure}[t]{0.15\textwidth}
        \centering
        \includegraphics[width=1\textwidth]{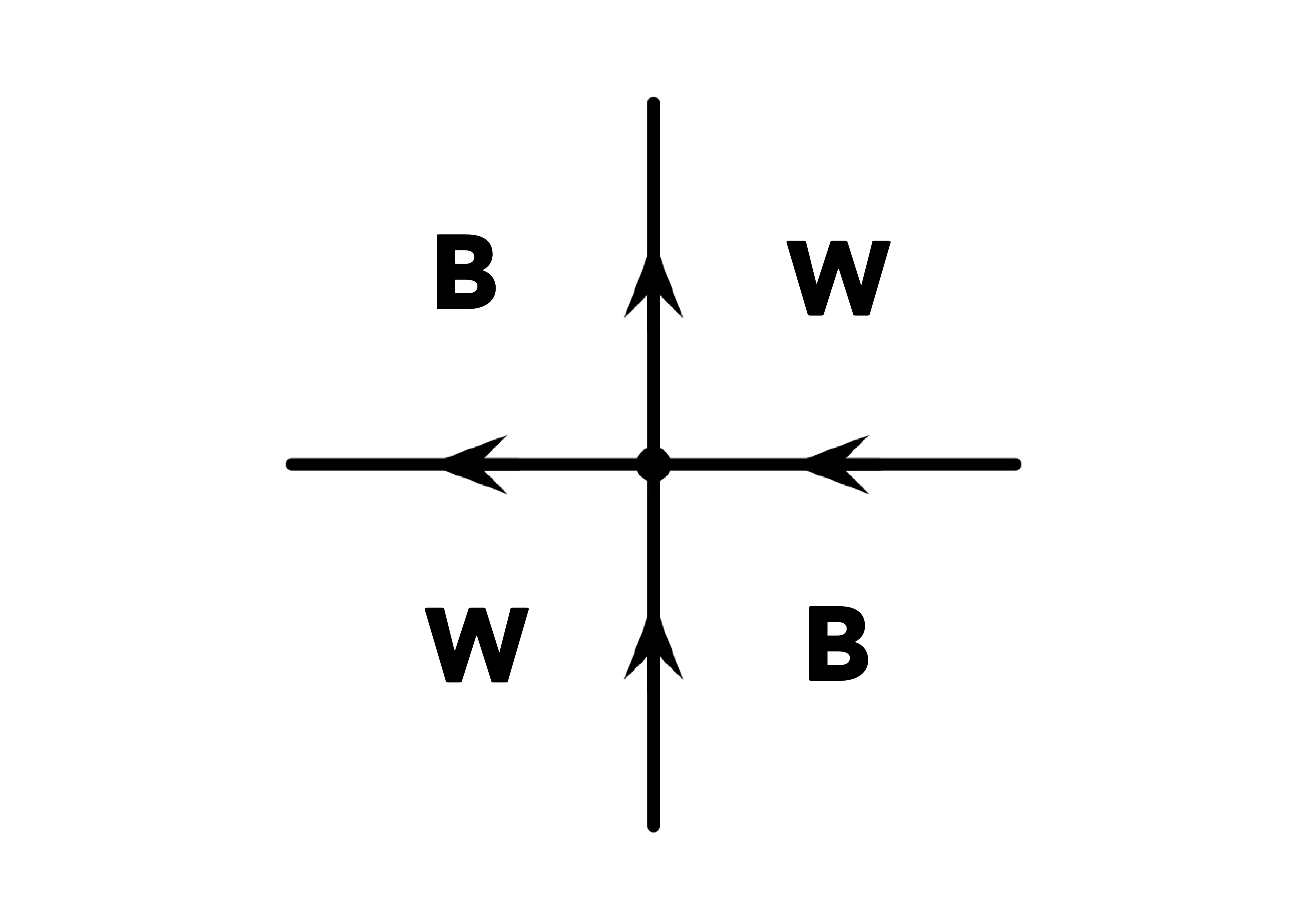}
        \subcaption*{4}
    \end{subfigure}
    \begin{subfigure}[t]{0.15\textwidth}
        \centering
        \includegraphics[width=1\textwidth]{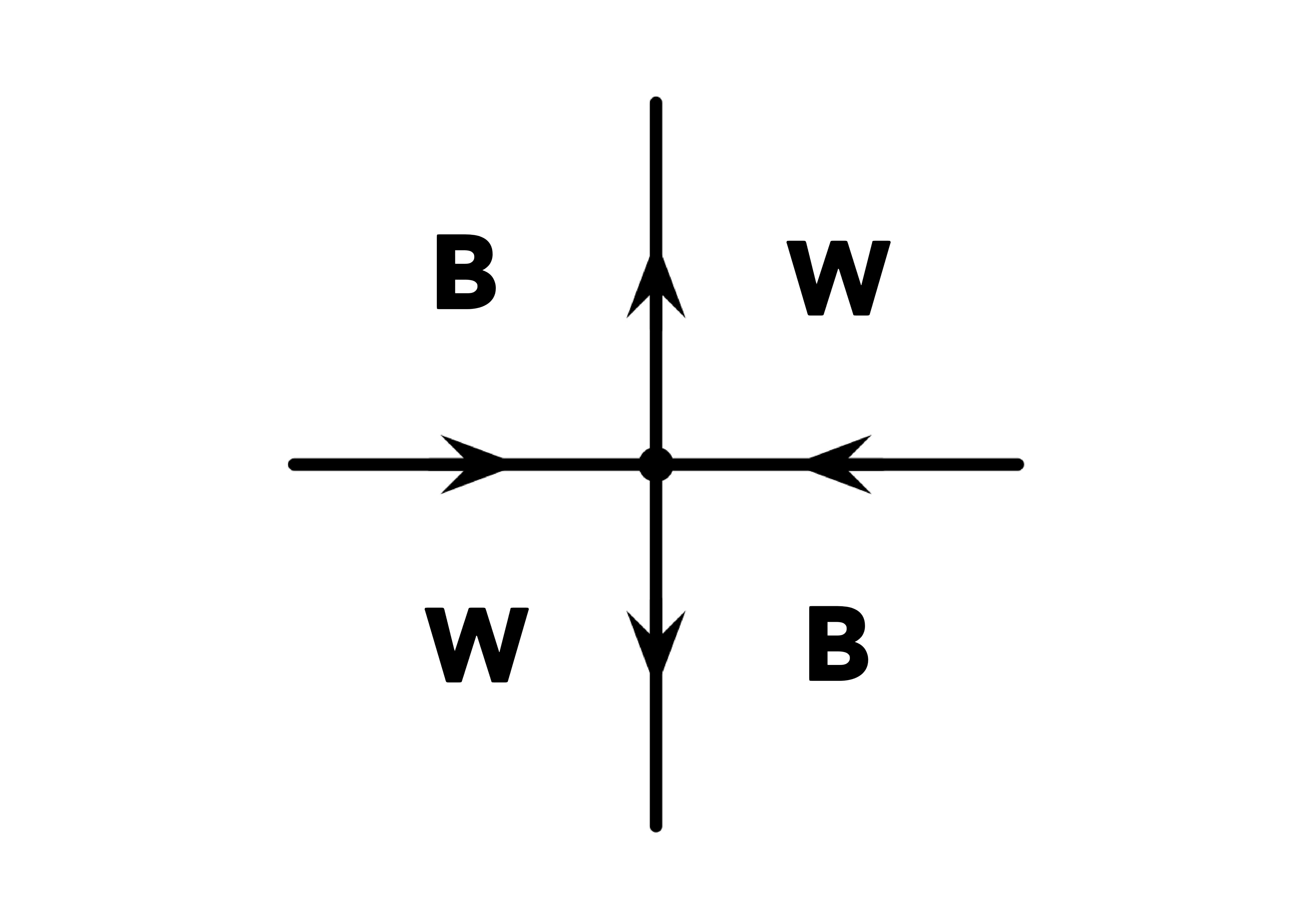}
        \subcaption*{5}
    \end{subfigure}
    \begin{subfigure}[t]{0.15\textwidth}
        \centering
        \includegraphics[width=1\textwidth]{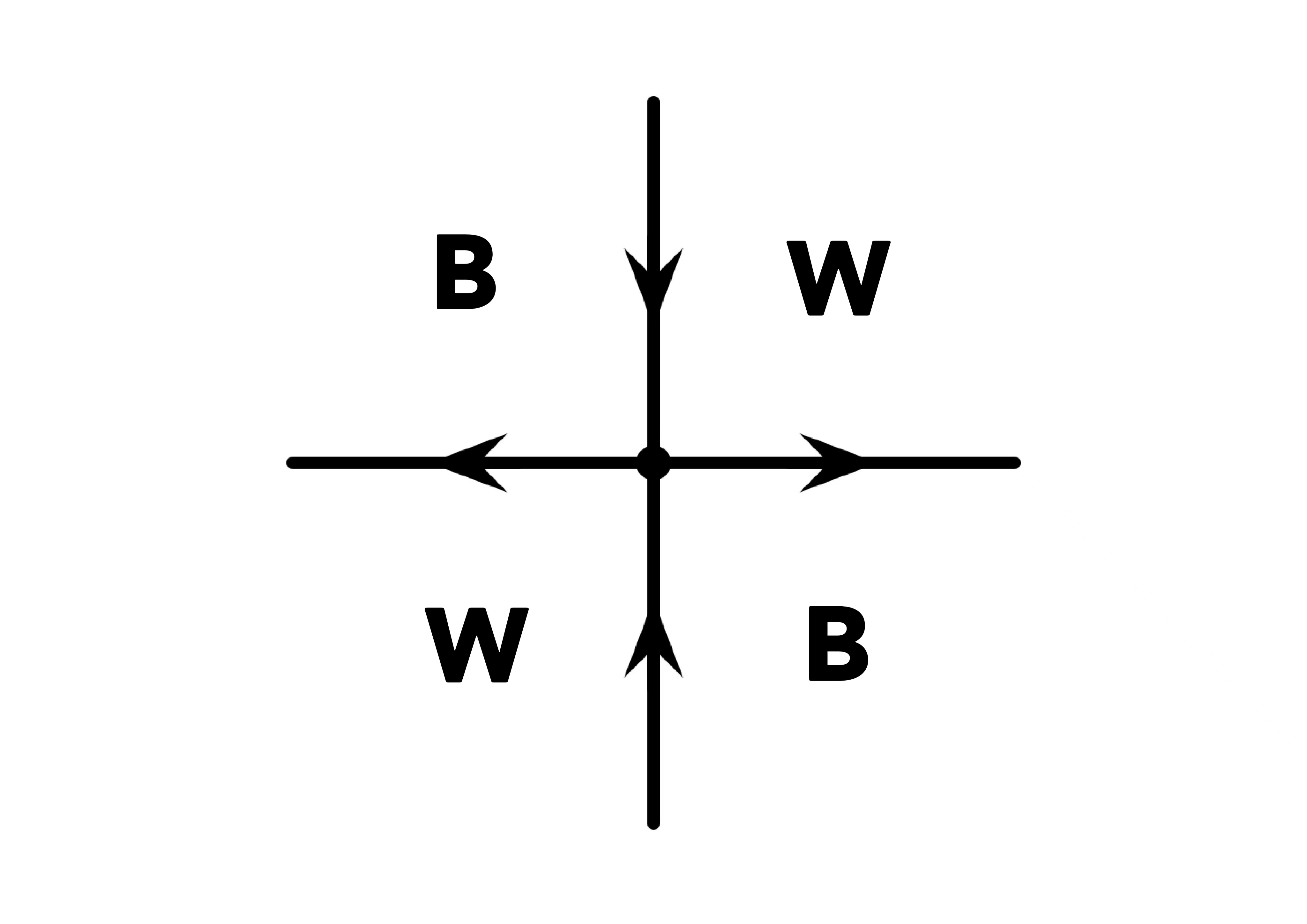}
        \subcaption*{6}
    \end{subfigure}
    \caption{Valid configurations of the planar six-vertex model}
    \label{planar six-type}
\end{figure}

\subsection{Sampling Algorithm}

Given a plane 4-regular graph $G$, a graph $H$ can be constructed as follows:
\begin{itemize}
    \item the vertices of $H$ are in 1-1 correspondence with black faces of $G$; every black face $F$ of $G$ contains exactly one vertex $v(F)$ of $H$; 
    \item the edges of $H$ are in 1-1 correspondence with vertices of $G$; given a vertex $x$ of $G$, let $F_1$ and $F_2$ be the two (possibly equal) black faces of $G$ incident to $x$, then the edge $e(x)$ of $H$ corresponding to $x$ is contained in $F_1\cap F_2\cap \{x\}$, contains $x$ and joins $v(F_1)$ and $v(F_2)$. 
\end{itemize}
We say that the graph $H$ is the graph of black face of $G$ and $G$ is the medial graph of $H$.
We denote the set of all black faces of $G$ by $\mathcal{B}$ and $H=( \mathcal{B}, E(H) )$.
Let $n=|V|$ and $k =|\mathcal{B}|$, then we have $|E|=2n$ and $|E(H)|=n$.

In a valid configuration of the planar four-vertex model, the directions of all edges on a black face form a directed circuit that can be either clockwise or counterclockwise.
Therefore, instead of assigning values of $0$ or $1$ to all edges on each black face, we can assign values of $0$ or $1$ to the black faces themselves.
If the two incident black faces of a vertex have the same direction, the vertex has weight $\beta$; otherwise, it has weight $1$.

Thus, the four-vertex model on $G$ can be reduced to a spin system on $H$.
The configurations of the spin system are $2^k$ possible assignments $\sigma :\mathcal{B}\to \{0,1\}$ of states to vertices.
The edges in $E(H)$ are labelled by the constraint function $f^\prime$ which has the constraint matrix $M(f^\prime)=
\left[ {\begin{array}{*{20}{c}}
    \beta&1\\
    1&\beta
\end{array}} \right]$.

The weight of a configuration $\sigma$ is given by
\begin{equation*}
    w(\sigma)=\prod_{(u,v)\in E(H)} f^\prime_{\sigma(u),\sigma(v)} = \prod_{(u,v)\in E(H)} \beta ^{{\bf I}(\sigma(u)=\sigma(v))}.
\end{equation*}
The partition function is given by $Z(H)=\sum_{\sigma}w(\sigma)$.
Notice that there might be multiple edges or self-loops in $H$.
The influence of self-loops to the partition function can be ignored since their constraint functions always take the value $\beta$.
We use the local constraint functions to replace the multiple edges.
In details, if there are multiple edges between $u,v\in\mathcal{B}$, we replace them with only one edge $e$ and label a local constraint function $M(f_e)=
\left[ {\begin{array}{*{20}{c}}
    \beta_e&1\\
    1&\beta_e
\end{array}} \right]$ where $\beta_e=\beta^{\#\text{multiple edges between } u,v}$ on it.
Denote the new edge set by $E^*$.
The weight of a configuration $\sigma$ can be rewritten as 
\begin{equation*}
    w(\sigma) = \prod_{(u,v)\in E^*} \beta_{(u,v)} ^{{\bf I}(\sigma(u)=\sigma(v))}.
\end{equation*}
The partition function can be rewritten as
\begin{equation}\label{equ:partition_Ising}
    Z_{4V} = \sum_{\sigma :\mathcal{B}\to \{0,1\}}\prod_{(u,v)\in E^*} \beta_{(u,v)} ^{{\bf I}(\sigma(u)=\sigma(v))}.
\end{equation}

Thus, we go back to the problem we have discussed in \cref{worm process}.
If $\beta>1$, the worm process can directly be applied to this situation, and this yields an FPAUS for the planar four-vertex model.
The planar structure of $G$ enables several optimizations.
By Euler's formula, the number of faces of $G$ is $n+2$, and there are at least two white faces in the graph.
Therefore, we have $k\le n$.
Additionally, the number of edges in $H$ is limited, as $|E^*|\le |E(H)|=n$.
Thus, we can state the following theorem:

\begin{theorem}
    For the planar four-vertex model $Z_{4v}(G;\beta,1)$ with $\beta>1$ under the canonical labeling, the mixing time of the worm process $\tau_\varepsilon(P)\le 4n^7(\frac{\ln 2\varepsilon^{-1}}{n}+\ln\frac{\beta_{\min}+1}{2(\beta_{\min}-1)})$ where $\beta_{\min}=\min_{e\in E^*}{\beta_e}$.
\end{theorem}

\bibliographystyle{plain}
\bibliography{myRef}

\begin{thebibliography}{10}

\bibitem{M.Backens/2018/Complex_Holant}
Miriam Backens.
\newblock A complete dichotomy for complex-valued {Holant}{\^{}}c.
\newblock In Ioannis Chatzigiannakis, Christos Kaklamanis, D{\'{a}}niel Marx,
  and Donald Sannella, editors, {\em 45th International Colloquium on Automata,
  Languages, and Programming, {ICALP} 2018, July 9-13, 2018, Prague, Czech
  Republic}, volume 107 of {\em LIPIcs}, pages 12:1--12:14. Schloss Dagstuhl -
  Leibniz-Zentrum f{\"{u}}r Informatik, 2018.

\bibitem{Baxter/book}
R.~J. Baxter.
\newblock {\em Exactly Solved Models in Statistical Mechanics}.
\newblock Academic Press Inc., San Diego, CA, USA, 1982.

\bibitem{Baxter/1972/Partition_Eight_Vertex}
Rodney~J Baxter.
\newblock Partition function of the eight-vertex lattice model.
\newblock {\em Annals of Physics}, 70(1):193--228, 1972.

\bibitem{Bogoliubov/2019/Four_Vertex_Partition}
Nikolay Bogoliubov and Cyril Malyshev.
\newblock The partition function of the four-vertex model in inhomogeneous
  external field and trace statistics.
\newblock {\em Journal of Physics A: Mathematical and Theoretical},
  52(49):495002, 2019.

\bibitem{Bogoliubov/2008/Four_Vertex_Tilings}
NM~Bogoliubov.
\newblock Four-vertex model and random tilings.
\newblock {\em Theoretical and Mathematical Physics}, 155(1):523--535, 2008.

\bibitem{Bulatov/2013/Complexity_Counting_CSP}
Andrei~A. Bulatov.
\newblock The complexity of the counting constraint satisfaction problem.
\newblock {\em J. {ACM}}, 60(5):34:1--34:41, 2013.

\bibitem{Bulatov/Dyer/Goldberg/Jalsenius/Jerrum/Richerby/2012/Complexity_Counting_CSP}
Andrei~A. Bulatov, Martin~E. Dyer, Leslie~Ann Goldberg, Markus Jalsenius, Mark
  Jerrum, and David Richerby.
\newblock The complexity of weighted and unweighted {\#}{CSP}.
\newblock {\em J. Comput. Syst. Sci.}, 78(2):681--688, 2012.

\bibitem{J.Cai/X.Chen/2017/Complexity_Complex_Counting_CSP}
Jin{-}Yi Cai and Xi~Chen.
\newblock Complexity of counting {CSP} with complex weights.
\newblock {\em J. {ACM}}, 64(3):19:1--19:39, 2017.

\bibitem{J.Cai/X.Chen/P.Lu/2013/GH_Complex}
Jin{-}Yi Cai, Xi~Chen, and Pinyan Lu.
\newblock Graph homomorphisms with complex values: {A} dichotomy theorem.
\newblock {\em {SIAM} J. Comput.}, 42(3):924--1029, 2013.

\bibitem{J.Cai/Z.Fu/S.Shao/2021/Planar_Six_Vertex}
Jin{-}Yi Cai, Zhiguo Fu, and Shuai Shao.
\newblock New planar {P}-time computable six-vertex models and a complete
  complexity classification.
\newblock In D{\'{a}}niel Marx, editor, {\em Proceedings of the 2021 {ACM-SIAM}
  Symposium on Discrete Algorithms, {SODA} 2021, Virtual Conference, January 10
  - 13, 2021}, pages 1535--1547. {SIAM}, 2021.

\bibitem{J.Cai/Z.Fu/M.Xia/2018/Complexity_Six_Vertex}
Jin{-}Yi Cai, Zhiguo Fu, and Mingji Xia.
\newblock Complexity classification of the six-vertex model.
\newblock {\em Inf. Comput.}, 259(Part):130--141, 2018.

\bibitem{J.Cai/H.Guo/T.Williams/2016/Vanishing_Signatures}
Jin{-}Yi Cai, Heng Guo, and Tyson Williams.
\newblock A complete dichotomy rises from the capture of vanishing signatures.
\newblock {\em {SIAM} J. Comput.}, 45(5):1671--1728, 2016.

\bibitem{Cai/Liu/20/8V_PM}
Jin-Yi Cai and Tianyu Liu.
\newblock {Counting Perfect Matchings and the Eight-Vertex Model}.
\newblock In Artur Czumaj, Anuj Dawar, and Emanuela Merelli, editors, {\em 47th
  International Colloquium on Automata, Languages, and Programming (ICALP
  2020)}, volume 168 of {\em Leibniz International Proceedings in Informatics
  (LIPIcs)}, pages 23:1--23:18, Dagstuhl, Germany, 2020. Schloss
  Dagstuhl--Leibniz-Zentrum f{\"u}r Informatik.

\bibitem{Cai/Liu/2020/FPRAS_Eight_Vertex}
Jin{-}Yi Cai and Tianyu Liu.
\newblock {FPRAS} via {MCMC} where it mixes torpidly (and very little effort).
\newblock {\em CoRR}, abs/2010.05425, 2020.

\bibitem{Cai/Liu/21/FPTAS_Vertex}
Jin{-}Yi Cai and Tianyu Liu.
\newblock An {FPTAS} for the square lattice six-vertex and eight-vertex models
  at low temperatures.
\newblock In D{\'{a}}niel Marx, editor, {\em Proceedings of the 2021 {ACM-SIAM}
  Symposium on Discrete Algorithms, {SODA} 2021, Virtual Conference, January 10
  - 13, 2021}, pages 1520--1534. {SIAM}, 2021.

\bibitem{J.Cai/T.Liu/P.Lu/2019/Approximability_Six_Vertex}
Jin{-}Yi Cai, Tianyu Liu, and Pinyan Lu.
\newblock Approximability of the six-vertex model.
\newblock In Timothy~M. Chan, editor, {\em Proceedings of the Thirtieth Annual
  {ACM-SIAM} Symposium on Discrete Algorithms, {SODA} 2019, San Diego,
  California, USA, January 6-9, 2019}, pages 2248--2261. {SIAM}, 2019.

\bibitem{CaiLLY20/Eight}
Jin{-}Yi Cai, Tianyu Liu, Pinyan Lu, and Jing Yu.
\newblock Approximability of the eight-vertex model.
\newblock In Shubhangi Saraf, editor, {\em 35th Computational Complexity
  Conference, {CCC} 2020, July 28-31, 2020, Saarbr{\"{u}}cken, Germany (Virtual
  Conference)}, volume 169 of {\em LIPIcs}, pages 4:1--4:18. Schloss Dagstuhl -
  Leibniz-Zentrum f{\"{u}}r Informatik, 2020.

\bibitem{Collevecchio/Garoni/Hyndman/Tokarev/2016/Worm_Raipd_Mixing}
Andrea Collevecchio, Timothy~M Garoni, Timothy Hyndman, and Daniel Tokarev.
\newblock The worm process for the {Ising} model is rapidly mixing.
\newblock {\em Journal of Statistical Physics}, 164(5):1082--1102, 2016.

\bibitem{Dyer/Richerby/2013/Dichotomy_Counting_CSP}
Martin~E. Dyer and David Richerby.
\newblock An effective dichotomy for the counting constraint satisfaction
  problem.
\newblock {\em {SIAM} J. Comput.}, 42(3):1245--1274, 2013.

\bibitem{DBLP:conf/approx/FahrbachR19}
Matthew Fahrbach and Dana Randall.
\newblock Slow mixing of glauber dynamics for the six-vertex model in the
  ordered phases.
\newblock In Dimitris Achlioptas and L{\'{a}}szl{\'{o}}~A. V{\'{e}}gh, editors,
  {\em Approximation, Randomization, and Combinatorial Optimization. Algorithms
  and Techniques, {APPROX/RANDOM} 2019, September 20-22, 2019, Massachusetts
  Institute of Technology, Cambridge, MA, {USA}}, volume 145 of {\em LIPIcs},
  pages 37:1--37:20. Schloss Dagstuhl - Leibniz-Zentrum f{\"{u}}r Informatik,
  2019.

\bibitem{Goldberg/Grohe/Jerrum/Thurley/2010/Complexity_Dichotomy_Mixed_Signs}
Leslie~Ann Goldberg, Martin Grohe, Mark Jerrum, and Marc Thurley.
\newblock A complexity dichotomy for partition functions with mixed signs.
\newblock {\em {SIAM} J. Comput.}, 39(7):3336--3402, 2010.

\bibitem{Goldberg/Mark/2012/Approximating_Potts}
Leslie~Ann Goldberg and Mark Jerrum.
\newblock Approximating the partition function of the ferromagnetic potts
  model.
\newblock {\em Journal of the ACM (JACM)}, 59(5):1--31, 2012.

\bibitem{L.Goldberg/R.Martin/M.Paterson/2004/Sample_3coloring}
Leslie~Ann Goldberg, Russell~A. Martin, and Mike Paterson.
\newblock Random sampling of 3-colorings in $\mathbb{Z}^2$.
\newblock {\em Random Struct. Algorithms}, 24(3):279--302, 2004.

\bibitem{Guo/Jerrum/2017/Cluster_Raipd_Mixing}
Heng Guo and Mark Jerrum.
\newblock Random cluster dynamics for the {Ising} model is rapidly mixing.
\newblock In Philip~N. Klein, editor, {\em Proceedings of the Twenty-Eighth
  Annual {ACM-SIAM} Symposium on Discrete Algorithms, {SODA} 2017, Barcelona,
  Spain, Hotel Porta Fira, January 16-19}, pages 1818--1827. {SIAM}, 2017.

\bibitem{Heng_Guo/Williams/2013/Complexity_Counting_CSP_Complex}
Heng Guo and Tyson Williams.
\newblock The complexity of planar {Boolean} {\#}{CSP} with complex weights.
\newblock In Fedor~V. Fomin, Rusins Freivalds, Marta~Z. Kwiatkowska, and David
  Peleg, editors, {\em Automata, Languages, and Programming - 40th
  International Colloquium, 
  Proceedings, Part {I}}, volume 7965 of {\em Lecture Notes in Computer
  Science}, pages 516--527. Springer, 2013.

\bibitem{L.Huang/P.Lu/C.Zhang/2016/Canonical_Path}
Lingxiao Huang, Pinyan Lu, and Chihao Zhang.
\newblock Canonical paths for {MCMC:} from art to science.
\newblock In Robert Krauthgamer, editor, {\em Proceedings of the Twenty-Seventh
  Annual {ACM-SIAM} Symposium on Discrete Algorithms, {SODA} 2016, Arlington,
  VA, USA, January 10-12, 2016}, pages 514--527. {SIAM}, 2016.

\bibitem{S.Huang/P.Lu/2016/Real_Holant}
Sangxia Huang and Pinyan Lu.
\newblock A dichotomy for real weighted {Holant} problems.
\newblock {\em Comput. Complex.}, 25(1):255--304, 2016.

\bibitem{M.Jerrum/A.Sinclair/1989/Approximating_Permanent}
Mark Jerrum and Alistair Sinclair.
\newblock Approximating the permanent.
\newblock {\em {SIAM} J. Comput.}, 18(6):1149--1178, 1989.

\bibitem{Jerrum/Sinclair/1993/Approximating_Ising}
Mark Jerrum and Alistair Sinclair.
\newblock Polynomial-time approximation algorithms for the {Ising} model.
\newblock {\em {SIAM} J. Comput.}, 22(5):1087--1116, 1993.

\bibitem{M.Jerrum/A.Sinclair/E.Vigoda/2004/Approximation_Permanent_Nonnegative}
Mark Jerrum, Alistair Sinclair, and Eric Vigoda.
\newblock A polynomial-time approximation algorithm for the permanent of a
  matrix with nonnegative entries.
\newblock {\em J. {ACM}}, 51(4):671--697, 2004.

\bibitem{Jerrum/Valiant/Vazirani/1986/counting_sampling}
Mark~R. Jerrum, Leslie~G. Valiant, and Vijay~V. Vazirani.
\newblock Random generation of combinatorial structures from a uniform
  distribution.
\newblock {\em Theoretical Computer Science}, 43(Supplement C):169 -- 188,
  1986.

\bibitem{Kasteleyn/1967/Graph_Crystal}
Pieter Kasteleyn.
\newblock Graph theory and crystal physics.
\newblock {\em Graph theory and theoretical physics}, pages 43--110, 1967.

\bibitem{Kasteleyn/1961/Dimer}
Pieter~W Kasteleyn.
\newblock The statistics of dimers on a lattice: I. the number of dimer
  arrangements on a quadratic lattice.
\newblock {\em Physica}, 27(12):1209--1225, 1961.

\bibitem{Kasteleyn/1963/Dimer}
Pieter~W Kasteleyn.
\newblock Dimer statistics and phase transitions.
\newblock {\em Journal of Mathematical Physics}, 4(2):287--293, 1963.

\bibitem{Kuperberg/1996/Another_ASM}
Greg Kuperberg.
\newblock Another proof of the alternative-sign matrix conjecture.
\newblock {\em International Mathematics Research Notices}, 1996(3):139--150,
  1996.

\bibitem{Las/1988/Evaluation_Tutte}
Michel Las~Vergnas.
\newblock On the evaluation at (3, 3) of the tutte polynomial of a graph.
\newblock {\em Journal of Combinatorial Theory, Series B}, 45(3):367--372,
  1988.

\bibitem{J.Lin/H.Wang/2018/Boolean_Holant}
Jiabao Lin and Hanpin Wang.
\newblock The complexity of {Boolean} {Holant} problems with nonnegative
  weights.
\newblock {\em {SIAM} J. Comput.}, 47(3):798--828, 2018.

\bibitem{M.Luby/D.Randall/A.Sinclair/2001/MC_Planar_Lattice}
Michael Luby, Dana Randall, and Alistair Sinclair.
\newblock Markov chain algorithms for planar lattice structures.
\newblock {\em {SIAM} J. Comput.}, 31(1):167--192, 2001.

\bibitem{C.McQuillan/2013/Winding}
Colin McQuillan.
\newblock Approximating {Holant} problems by winding.
\newblock {\em CoRR}, abs/1301.2880, 2013.

\bibitem{M.Mihail/P.Winkler/1996/EO}
Milena Mihail and Peter Winkler.
\newblock On the number of eulerian orientations of a graph.
\newblock {\em Algorithmica}, 16(4/5):402--414, 1996.

\bibitem{Pauling/1935/Ice_type}
Linus Pauling.
\newblock The structure and entropy of ice and of other crystals with some
  randomness of atomic arrangement.
\newblock {\em Journal of the American Chemical Society}, 57(12):2680--2684,
  1935.

\bibitem{Prokof/Boris/2001/Worm_Algorithm}
Nikolay Prokof'ev and Boris Svistunov.
\newblock Worm algorithms for classical statistical models.
\newblock {\em Physical review letters}, 87(16):160601, 2001.

\bibitem{D.Randall/P.Tetali/1998/Glauber_Comparison}
Dana Randall and Prasad Tetali.
\newblock Analyzing {Glauber} dynamics by comparison of {Markov} chains.
\newblock In Claudio~L. Lucchesi and Arnaldo~V. Moura, editors, {\em {LATIN}
  '98: Theoretical Informatics, Third Latin American Symposium, Campinas,
  Brazil, April, 20-24, 1998, Proceedings}, volume 1380 of {\em Lecture Notes
  in Computer Science}, pages 292--304. Springer, 1998.

\bibitem{Schweinsberg/2002/Relaxation_Time}
Jason Schweinsberg.
\newblock An ${O}(n^2)$ bound for the relaxation time of a {Markov} chain on
  cladograms.
\newblock {\em Random Struct. Algorithms}, 20(1):59--70, 2002.

\bibitem{A.Sinclair/1992/Mixing_FLow}
Alistair Sinclair.
\newblock Improved bounds for mixing rates of {Markov} chains and
  multicommodity flow.
\newblock {\em Combinatorics, probability and Computing}, 1(4):351--370, 1992.

\bibitem{Temperley/Fisher/1961/Dimer}
Harold~NV Temperley and Michael~E Fisher.
\newblock Dimer problem in statistical mechanics-an exact result.
\newblock {\em Philosophical Magazine}, 6(68):1061--1063, 1961.

\bibitem{Tutte/1954/Chromatic_Polynomials}
William~Thomas Tutte.
\newblock A contribution to the theory of chromatic polynomials.
\newblock {\em Canadian journal of mathematics}, 6:80--91, 1954.

\bibitem{Valiant/2004/Holographic_Algorithms}
Leslie~G. Valiant.
\newblock Holographic algorithms (extended abstract).
\newblock In {\em 45th Symposium on Foundations of Computer Science {(FOCS}
  2004), 17-19 October 2004, Rome, Italy, Proceedings}, pages 306--315. {IEEE}
  Computer Society, 2004.

\bibitem{Welsh/1969/Matroids}
Dominic~JA Welsh.
\newblock Euler and bipartite matroids.
\newblock {\em Journal of Combinatorial Theory}, 6(4):375--377, 1969.

\end{thebibliography}
\end{document}